\documentclass[a4paper,11pt,oneside]{amsart}

\pdfoutput=1

\usepackage{amsaddr}     
\usepackage[textheight=24cm,textwidth=16cm]{geometry}
\usepackage{appendix}
\usepackage[pdftex]{graphicx}

\usepackage{hyperref}

\usepackage[T1]{fontenc}

\usepackage{natbib}

\begin{document}

\newcommand{\esp}{\mathsf{E}}
\newcommand{\med}{\mathsf{med}}
\newcommand{\var}{\mathsf{var}}
\newcommand{\std}{\mathsf{std}}
\newcommand{\cov}{\mathsf{cov}}
\newcommand{\Err}{\text{Err}}
\newcommand{\bd}{\boldsymbol}
\newcommand{\etr}{\text{etr}}
\newcommand{\Tr}{\text{tr}}
\newcommand{\Hzero}{\textsc{H}_0}
\newcommand{\Hone}{\textsc{H}_1}
\newcommand{\Hi}{\textsc{H}_i}
\newcommand{\LR}{\textsc{LR}}
\newcommand{\GLR}{\textsc{GLR}(x)}
\newcommand{\PLR}{\textsc{PLR}}
\newcommand{\POR}{\textsc{POR}(x)}
\newcommand{\pOR}{\textsc{pOR}}
\newcommand{\BF}{\textsc{BF}(x)}
\newcommand{\FBF}{\textsc{FBF}}
\newcommand{\pval}{p_{\text{val}}}
\newcommand{\PFA}{\textsc{PFA}}
\newcommand{\PD}{\textsc{PD}}
\newcommand{\R}{\mathbb R}
\newcommand{\G}{\mathcal G}
\newcommand{\X}{\mathcal X}
\newcommand{\Le}{\mathcal L}
\newcommand{\Ha}{\mathcal H}
\newcommand{\ml}{\hat{\theta}_{\text{ML}}(x)}
\newcommand{\MMSE}{\hat{\theta}_{\text{MMSE}}(x)}
\newcommand{\MAP}{\hat{\theta}_{\text{MAP}}(x)}
\newcommand{\rmin}{r_{\text{min}}}
\newcommand{\rmax}{r_{\text{max}}}
\newcommand{\vol}{\text{vol}}
\renewcommand{\Pr}{\text{Pr}}

\newtheorem{theorem}{Theorem}
\newtheorem{proposition}{Proposition}
\newtheorem{lemma}{Lemma}
\newtheorem{corollary}{Corollary}
\newtheorem{remark}{Remark}

\title[Generalizations related to the Posterior distribution of the LR]
{Generalizations related to hypothesis testing with the Posterior distribution of the Likelihood Ratio.}

\begin{abstract}
The Posterior distribution of the Likelihood Ratio (PLR) is proposed by Dempster in 1974 for significance testing in the simple vs composite hypotheses case. In this hypotheses test case, classical frequentist and Bayesian hypotheses tests are irreconcilable, as emphasized by Lindley's paradox, Berger \& Selke in 1987 and many others. However, Dempster shows that the PLR (with inner threshold 1) is equal to the frequentist p-value in the simple Gaussian case. In 1997, Aitkin extends this result by adding a nuisance parameter and showing its asymptotic validity under more general distributions. Here we extend the reconciliation between the PLR and a frequentist p-value for a finite sample, through a framework analogous to the Stein's theorem frame in which a credible (Bayesian) domain is equal to a confidence (frequentist) domain. 

This general reconciliation result only concerns simple vs composite hypotheses testing. The measures proposed by Aitkin in 2010 and Evans in 1997 have interesting properties and extend Dempster's PLR but only by adding a nuisance parameter. Here we propose two extensions of the PLR concept to the general composite vs composite hypotheses test. The first extension can be defined for improper priors as soon as the posterior is proper. The second extension appears from a new Bayesian-type Neyman-Pearson lemma and emphasizes, from a Bayesian perspective, the role of the LR as a discrepancy variable for hypothesis testing.
\end{abstract}


\author{I. Smith}
\address{Laboratoire des Sciences du Climat et de l'Environnement ; IPSL-CNRS, France.\\
Universit\'e de Nice Sophia-Antipolis, CNRS, Observatoire de la C\^ote d'Azur, France.}
\email{zazoo@mac.com}

\author{A. Ferrari}
\address{Universit\'e de Nice Sophia-Antipolis, CNRS, Observatoire de la C\^ote d'Azur, France}
\email{andre.ferrari@unice.fr}
\thanks{Part of this work was published in \citet{smith14}}

\maketitle

\section{Introduction}


\subsection{Classical hypotheses test methodologies}
\label{sec_class_test}

Simple versus composite hypotheses testing is a general statistical issue in parametric modeling. It consists for a given observed dataset $x$ in choosing among the hypotheses
\begin{align}
\Hzero: \theta=\theta_0 ~~~~~ \Hone: \theta\in\Theta_1  
\end{align}
where the distribution of $x$ is characterized by the underlying unknown parameter $\theta$. Under the alternative hypothesis $\Hone$, $\theta$ takes a value different from the point $\theta_0$, and the uncertainty of $\theta$ is described by a prior probability density function $\pi_1(\theta)$ which is positive only for $\theta\in\Theta_1$. We assume that the data model $p(x|\theta)$ has the same expression under $\Hzero$ and $\Hone$.

To choose among $\Hzero$ and $\Hone$, a test statistic $T(x)$ (such as the Generalized Likelihood Ratio)
is generally compared to a threshold $\zeta$ and one decides to choose $\Hzero$ if $T(x)$ is greater than $\zeta$. If $\Hone$ is chosen whereas the true underlying $\theta$ was equal to $\theta_0$, a type I error is made in the decision.
Under the classical Neyman paradigm (see \citet{neyman33,neyman77}), the threshold $\zeta$ is chosen so that the 
probability of the type I error lies under (or is equal to) some fixed level $\alpha$, typically a 5\% error rate. Instead of inverting  this function, a p-value can be defined in order to serve as the test statistic to be directly compared to the 5\% level (\citet{lehmann05}):
\begin{align}
\pval(T(x_0)) = \Pr(T(x_0) < T(x)|\theta_0)   \label{eq_pval_f}
\end{align}
where $x_0$ is the observed dataset and $x$ the variable of integration. 
Note that with this notation, $\Hzero$ is rejected  when $\pval(T(x_0))$ is greater than some threshold.

On the Bayesian side, the test statistic classically used (\citet{robert07}) is the Bayes Factor (BF) defined by 
\begin{equation*}
\BF = \frac{p(x|\theta_0)}{\int d\theta~p(x|\theta)\pi_1(\theta)}   
\end{equation*}
Making a binary decision consists of choosing $\Hzero$ if $\BF$  is greater than some threshold, and the choice of the threshold is made in general by a straight interpretation of the BF. The Jeffreys's scale for example states that if the observed BF is between 10 and 100 there is a strong evidence in favor of $\Hzero$. The mere posterior probability $\Pr(\Hi|x)$ of an hypothesis may also be considered by itself.

A practical issue of the BF in the simple vs composite hypotheses test is that it is defined up to a multiplicative constant if the prior $\pi_1$ is improper\footnote{$\pi_1$ is called improper if its integral over $\Theta_1$ is infinite, which occurs if $\pi_1$ is constant over an unbounded domain for example.} even though the posterior distribution is proper. Partial BFs account for this issue by somehow using part of the data to update the prior into a proper posterior, and then use this posterior as the prior for the rest of the data. The most simply defined Partial BF is the Fractional BF (FBF) proposed by \citet{ohagan95}.

A related and more fundamental issue is Lindley's paradox, initially studied by \citet{jeffreys61} and called a paradox by \citet{lindley57}, which shows among others that, when testing a simple vs a composite hypothesis, the null hypothesis $\Hzero$ is too highly favoured against $\Hone$ for a natural diffuse prior under $\Theta_1$. More precisely, for example in the test of the mean of a Gaussian likelihood, the p-value $|x|$ defines the uniformly most powerful test, which is a very strong optimal property even according to at least part of the Bayesian community. However, for a fixed prior and some dataset $x$ that adjusts so that the associated classical p-value remains fixed (so that the evidence for $\Hzero$ shall not change), $\Pr(\Hzero|x)/\Pr(\Hone|x)$ tends to 1 as the sample size increases. This issue, intensively discussed and developed (see \citet{tsao06} for a quite recent study), is consensually considered as a real trouble by a quite large part of the community. Unlike the BF, other tests like the FBF or the \citet{bernardo10} test do not suffer from this problem in Lindley's frame. Other ideas have been developed which prevent Lindley's frame from occurring, avoiding troubles for the BF. \citet{berger87} for example argue that testing a simple hypothesis is an unreasonable question. Some other references will be given in the section \ref{sec_prev_temp}.

Among many frequentist and Bayesian p-values (several are listed by \citet{robins00}), the next most classical Bayesian-type hypotheses test statistic is the posterior predictive p-value, highlighted by \citet{meng94}. Unlike the BF which only integrates over the parameter space $\Theta$, the posterior predictive p-value integrates over the data space $\X$, like frequentist p-values. But unlike the frequentist p-value which integrates under the frequentist likelihood $p(x|\theta_0)$, it integrates under the predictive likelihood $p(x^{\text{pred}}|x_0) = \int d\theta~p(x^{\text{pred}}|\theta)\pi(\theta|x_0)$ where $x_0$ is the observed dataset. In a frequentist p-value only a statistic (ie a function of $x$ only) can define the domain of integration. On the contrary, in the posterior predictive p-value, a discrepancy variable (function of both $x$ and $\theta$) can be used to define the domain of integration. Note that the choice of the discrepancy variable to use there remains an issue.

Although a bit less classical, the approach of \citet{evans97} needs to be introduced because the tool and some of its properties are interesting and closely related to the ones derived in this paper. In the simple vs composite test case presented up to now, the tool proposed by \citet{evans97} and the ones studied in this paper are even mathematically equal. But the tool proposed by \citet{evans97} is defined to test more generally $\Hzero: \Psi(\theta) = \psi_0$ for a parameter of interest $\psi = \Psi(\theta)$. The test statistic consists of measuring the Observed Relative Surprise (ORS) related to the hypotheses by computing:
\begin{align}
\textsc{ORS}(x) = \Pr\left( \frac{\pi_\Psi(\Psi(\theta) | x)}{\pi_\Psi(\Psi(\theta))} \ge \frac{\pi_\Psi(\psi_0 | x)}{\pi_\Psi(\psi_0)}  \bigg{|} x\right)  \label{ORS}
\end{align}
The relative belief ratio of $\psi$ defined by $\textsc{RB}(\psi) = \pi_\Psi(\Psi(\theta) | x)\pi_\Psi(\Psi(\theta))^{-1}$ is measuring the change in belief in $\psi$ being the true value of $\Psi(\theta)$ from \textit{a priori} to \textit{a posteriori}. So if $\textsc{RB}(\psi_0)>1$ we have evidence in favor of $\Hzero$. Relative belief ratios are discussed in \citet{baskurt13} where $\textsc{RB}(\psi_0)$ is presented as the evidence for or against $\Hzero$ and (\ref{ORS}) is presented as a measure of the reliability of this evidence. This leads to a possible resolution of Lindley's paradox as the relative belief ratio can be large and ORS small without contradiction. See Example 4 of \citet{baskurt13} and note that \citet{evans97} shows that ORS converges to the classical p-value as the prior becomes more diffuse in this example.

\subsection{Posterior distribution of the Likelihood Ratio (PLR)}
\label{sec_intro_plr}

Let's focus again on the simple vs composite hypothesis test. Contrary to the posterior predictive p-value, the Posterior distribution of the Likelihood Ratio (PLR) does not integrate over some data which are unobserved, but only integrates over $\Theta$. It still conditions upon the only observed variable, namely $x_0$, like for the BF, but on a domain defined from a divergence variable, like the posterior predictive p-value. This statistic proposed by \citet{dempster74} is defined by
\begin{equation}
\PLR(x, \zeta) = \Pr\big(\LR(x,\theta) \le \zeta \big| x\big)  \label{def_PLR}
\end{equation}
where $\LR(x,.)$ is the Likelihood Ratio
\begin{equation*}
\LR(x,\theta) = \frac{p(x|\theta_0)}{p(x|\theta)}  ~~~~~~~~~~~~ \theta\in\Theta_1 
\end{equation*}

Since $\theta$ is random, the deterministic function $\LR(x,.)$ evaluated at the random variable $\theta$ becomes naturally random with some posterior distribution characterized by its cumulative distribution, the PLR. As emphasized by \citet{birnbaum62}, \citet{dempster74} and \citet{royall97}, the threshold $\zeta$ which compares the original likelihoods under $\Hzero$ and under $\Hone$ is directly interpretable and can be chosen the same way an error level $\alpha$ is chosen in the Neyman-Pearson paradigm. ``$\PLR(x,1) = 0.1$'' for example reads ``The probability that the likelihood of $\theta_1$ is more than the likelihood of $\theta_0$ is 0.1.''. 

The PLR can therefore be used for a binary decision, by fixing $\zeta$ and deciding to reject $\Hzero$ if $\PLR(x,\zeta)$ is greater than, say, 0.9. One can check if the binary decision is sensitive to the choice of both thresholds by making the test for several thresholds and see if the decision is different. In the extreme case, note that due to the nice definition of the PLR, one can simply display $\PLR(x,\zeta)$ as a function of $\zeta$ to get a broad view. The range of $\zeta$ under which $\PLR(x,\zeta)$ grows typically from 0.2 to 0.8 indicates if the decision for $\Hzero$ or $\Hone$ is clear, or not. As soon as the posterior can be sampled, these computations and graphs are very easy to display as will be explained later.

The PLR has been first proposed by \citet{dempster74, dempster97}, then studied especially by \citet{aitkin97} and \citet{aitkin10} but also used and analyzed by \citet{aitkin05, aitkin09}. As mentioned in the previous subsection, it turns out that the PLR is also closely related to the ORS proposed by \citet{evans97}, which generalizes the PLR. The PLR is also closely related to the e-value associated to the Full Bayesian Significance Test (FBST) from \citet{pereira99} and slightely revisited by \citet{borges07} which then somehow generalizes the PLR by adding a reference distribution on $\theta$, and by systematically dealing with the case where the null hypothesis domain $\Theta_0$ has a dimension less than $\Theta_1$ but which is not necessarily restricted to the point $\Theta_0 = \{\theta_0\}$. We do not list the results found by these different analyses, apart from some specifically mentioned ones.

The PLR turns out to be a natural Bayesian measure of evidence of the studied hypotheses since it involves only the posterior distribution of $\theta$ (no integral over $\mathcal X$) and the likelihood, claimed by \citet{birnbaum62}, \citet{royall97} and others, to be the only tool that can measure evidence. Unlike the BF, the PLR is well defined for an improper prior as soon as the posterior is proper, and is not subject to Lindley's paradox. It is also invariant under any isomorphic transformation of the $\X$ space and any transformation of the $\Theta$ space, as a consequence of being a mere function of the likelihood. These last properties were emphasized for example for the e-value associated to the FBST. 

The PLR also consists in a natural alternative to the BF in different regards. 
To start with, the PLR first compares (compares $p(x|\theta_0)$ and $p(x|\theta)$) and then integrates, whereas the BF first integrates and then compares (compares $p(x|\theta_0)$ and $\int d\theta~p(x|\theta)\pi(\theta)$). 
Second, \citet{newton94} and many others show that if the prior under $\Hone$ is proper, the BF is simply the posterior mean of the LR, ie the mean of the distribution described by the PLR\footnote{Alternatively, note that if we had defined the BF and LR with the alternative hypothesis at the numerator of these fractions, the BF would have been the prior mean of the LR.}. However a point estimate is in general not given alone but accompanied by an uncertainty indicator. \citet{smith10d} show that the posterior mean of the LR raised at some power is equal to the FBF introduced previously;
the mean of the PLR is given by the BF and its variance is easily related to the FBF.
However, \citet{smith10e} shows that the Generalized Likelihood Ratio bounds the support (values of $\zeta(x)$ for which $\PLR(x,\zeta)>0$) of the PLR and that at this lower bound the PLR in general starts by an infinite derivative. In addition to this theoretical result, numerical examples also indicate that the posterior density function of the LR is in general highly asymmetric. Therefore, the BF (point estimate of the LR) or any standard centered credible intervals do not appear to be relevant inferences about the LR seen as random variable. Instead, the same way the BF is to be thresholded, the actual information about $\LR(x,\theta)$ which seems to be relevant, and invariant under the transformation $\LR(x,\theta) \mapsto (\LR(x,\theta))^{-1}$, is to indicate its cumulative posterior distribution, which is precisely the PLR.

In practice, the PLR can be straightforwardly computed as soon as the posterior distribution $\pi(\theta|x)$ is sampled. Just obtain from a Monte Carlo Markov Chain (MCMC) algorithm an almost i.i.d. chain $\{\theta^{[1]},...,\theta^{[m]}\}$ from the posterior distribution $\pi(\theta|x)$ and compute $\LR(x,\theta^{[i]})$ for each sample. The resulting histogram sketches the posterior density of the LR and the plot of the empirical cumulative distribution of the LR chain sketches the PLR as a function of $\zeta$.

The PLR has been realistically and thouroughly applied by \citet{smith10e} to the detection of extra-solar planets from images acquired with the dedicated instrument SPHERE mounted on the Very Large Telescope. At this moment, only very finely simulated images were available. The PLR has been applied to two simulated datasets, one in which no extra-solar planet is present (dataset simulated under $\Hzero$) and the other in which an extra-solar planet is present ($\Hone$ dataset). 
Although the extra-solar planet is very dark ($10^6$ times less bright than the star it surrounds), close to the star (angular distance in the sky of $0.2$ arcseconds i.e. $6.10^{-5}$ degrees), and although only $2\times20$ images were used, thanks to the quality of the optical instruments and of the statistical model the detection and not detection were evident, with $\PLR(x,0.1) = 0.0$ for the dataset under $\Hzero$ and $\PLR(x,0.1) = 0.94$ for the dataset under $\Hone$. As studied by \citet{smith10e}, the statistical model and consecutive method are very satisfying compared to classical methods.  


\subsection{Problematics addressed here}

Despite its potential interest the PLR has not been extensively studied up to now. This paper aims at contributing in this investigating work by some new results. 

In the simple vs composite hypotheses test case, it turns out that the PLR plays a strong role in understanding the possible reconciliation between frequentist and Bayesian hypothesis testing. The PLR with inner threshold $\zeta=1$ is simply equal to some frequentist p-value for some ``likelihood - prior - hypotheses'' combinations. \citet{dempster74, aitkin97} have first noticed and highlighted this equivalence when testing the mean of a gaussian likelihood with a uniform prior. 

In the section \ref{sec_eq_plr_freq}, we extend the conditions of this equivalence result under a frame analogous to the one used to reconcile confidence and credible domains. The subsection \ref{sec_prev_temp} synthesizes the long quest of reconciliation between frequentist and Bayesian hypotheses tests, the subsection \ref{sec_equi} proves and discusses the reconciliation reached between the PLR and some frequentist p-value in such an invariant frame, the subsection \ref{sec_ex_pers} gives examples and perspectives, and the subsection \ref{sec_connect} discusses the connection between this reconciliation result and the one obtained between (frequentist) confidence domains and (Bayesian) credible regions.

\citet{aitkin97} and \citet{aitkin10} extended the PLR definition to an hypotheses test frame identical to the one presented at the end of the subsection \ref{sec_class_test}, namely $\Hzero: \Psi(\theta) = \psi_0$, also considered by \citet{evans97} and others. However, the PLR has not been yet generalized to the general composite vs composite hypotheses test. The generalization is somehow unnatural for a frequentist p-value because for a simple hypothesis $\Hzero:\theta=\theta_0$ the p-value is a frequentist probability conditioned on the fixed parameter $\theta_0$ (see equation \ref{eq_pval_f}), although a conditional probability cannot be defined on a composite set $\Theta_0$ if no probability distribution over $\Theta$ is used. By contrast, the PLR is reciprocally a probability conditioned upon the observed dataset $x_0$, and $x_0$ naturally remains fixed under a composite hypothesis. Therefore, the transition from simple to composite null hypothesis does not raise immediate obstacles for the PLR. However, a joint measure on the parameter spaces of both hypotheses is still required. 

The section \ref{sec_generalization} proposes and motivates two generalizations of the PLR. The mathematical expressions of the two extensions are simply given and rephrased in the subsection \ref{sec_ext_plr}. The first extension in particular enables the use of improper priors as soon as the posterior is proper. It can therefore be used in the subsection \ref{sec_examp} for the detection of precipitation change where an almost improper prior is to be used but leads to a proper posterior. On the other side, the second extension, made of two symmetrical probabilities, appears in a Bayesian version of the Neyman-Pearson lemma. As detailed in the subsection \ref{sec_bayes_neyman}, the two joint measures associated to no specific discrepancy variable lead through the lemma to the discrepancy variable $\LR(x_0,\theta)$.

A concluding discussion is proposed in the section \ref{sec_conclu}. The appendices essentially present the proofs of the mathematical results.


\section{Equivalence between the PLR and a frequentist p-value}
\label{sec_eq_plr_freq}

\subsection{Previous tentative reconciliations of frequentist and Bayesian tests}
\label{sec_prev_temp}

As introduced in the section \ref{sec_class_test}, Lindley's paradox presents a frame where $\Pr(\Hzero|x)$ (often thought as being \textit{the} Bayesian measure of evidence) may be expected to be equal to the frequentist p-value, but happens not to be. Also, the BF is not satisfying in the frame ``point null hypothesis $\Hzero$ and diffuse prior $\pi_1$''. This highlights the need for other Bayesian-type hypotheses tests, but also raises more generally the question of reconciliation between frequentist and Bayesian hypotheses tests. 

The conditions upon which frequentist (\citet{neyman77}) and Bayesian (\citet{jeffreys61}) answers agree is always of interest in order to understand the interpretation of the procedures and the limits of the two paradigms, somehow defined by what they are \textit{not}. 

A first approach to see when could frequentist and Bayesian hypotheses tests be unified consists of analyzing, for different hypotheses likelihoods and priors, when are the classical p-value and $\Pr(\Hzero|x)$ equal. These two concepts are to be compared because they both seem to handle only $\Hzero$ and in very simple ways, one from the frequentist the other from the Bayesian perspectives\footnote{Note however that $\Hone$ is implicitely taken into account through the marginal distribution of $x$ in $\Pr(\Hzero|x)$.}. It turns out that unlike for a composite null hypothesis (e.g. \citet{casella87}), for a point null hypothesis Lindley's paradox $\Pr(\Hzero|x) > \pval(x)$ always seems to hold. \citet{berger87b} in particular show that among very broad classes of priors $\Pr(\Hzero|x) > \pval(x)$ always holds for $\Pr(\Hzero)=0.5$. Also see the extensive list of references included. \citet{oh99} follows this analysis by studying the effect of the choice of $\Pr(\Hzero)$. 

Another approach consists of modifying the standard frequentist procedure and/or the standard Bayesian hypotheses test procedure, but still relying on the p-value and $\Pr(\Hzero|x)$, to see if they can then be made equivalent. \citet{berger87} for example study ``precise'' (concentrated) but not exactly ``point'' hypotheses, \citet{berger94} use frequentist p-values computed from a likelihood conditioned upon a set in which lies the observed dataset, not on the dataset itself, and define a non-decision domain in the BF test procedure. \citet{sellke01} advocate calibrating (\textit{rescaling}) the frequentist p-value to relate this new statistic to other test statistics. 

As already mentioned in the section \ref{sec_class_test}, one can also try to unify the p-value to Bayesian type statistics fully \textit{different} from the BF, to see when frequentist and Bayesian types hypotheses tests can be made equivalent. In particular, when \citet{dempster74} proposed to use the PLR, he also mentioned that when testing the mean of a normal distribution, the PLR is equal to the classical frequentist p-value when computed for a uniform prior and with inner parameter $\zeta=1$. This fundamental result was again emphasized by \citet{aitkin97} and \citet{dempster97}.

\citet{aitkin97}  asymptotically extended this result to any regular distribution, making use of the asymptotic convergence of a regular distribution towards a normal distribution. For any regular continuous distribution and a smooth prior, the PLR, with $\zeta=1$, tends asymptotically to the classical p-value. Also, with a nuisance parameter $\eta$ and still calling $\theta$ the tested parameter, he defines LR by $\LR(x,\theta,\eta) = p(x|\theta_0,\eta)/p(x|\theta,\eta)$, in which case under the same conditions as in the previous case the PLR is equal to a p-value. For a normal distribution, when testing the mean and considering the variance as a nuisance parameter, the result is also true for a finite sample.

\subsection{New reconciliation result}
\label{sec_equi}

The sets of conditions found by \citet{dempster74} and \citet{aitkin97} under which the PLR (with $\zeta=1$) is equal to a p-value are directly related to the test of the mean of a normal distribution under a uniform prior. The next subsection generalizes this exact finite-sample result under the frame of statistical invariance. As will be discussed at the end of the section, although the technical conditions derived here may be relaxed, it may be difficult to find, at least within the current statistical frame, a fundamentally more general frame of conditions for an equality between the PLR and a p-value to hold.

As presented in current classical textbooks in Bayesian statistics (\citet{berger85}, \citet{robert07}), invariance in statistics arises from the invariant Haar measure defined on some topological group. Throughout this subsection and the related appendices, we will use the notions and results synthesized by \citet{nachbin65} and \citet{eaton89}. The tools necessary to understand the result are introduced in the appendix 1.

In this frame, the PLR (given by an integral over the parameter space $\Theta$) can be reexpressed as an integral over the sample space $\mathcal X$, equal to a p-value for $\zeta=1$. In this subsection $x$ and $\theta$ denote random variables or variables of integration according to the context. 

First, for clarity, we give the equivalence between the PLR and a frequentist integral under the assumption that the sample space $\X$, the parameter space $\Theta$ and the transformations group $\G$ are isomorphic. 

\begin{theorem}
\label{th_lebesgue}
Call $\mathcal P_\Theta = \{p(.|\theta),\theta\in\Theta\}$ a family of probability densities with respect to the Lebesgue measure on $\mathcal X$, and call $\G$ a group acting on $\X$. Assume that $\mathcal P_\Theta$ is invariant under the action of the group $\G$ on $\mathcal X$ and note $\bar g\theta$ the induced action of the element $g\in\G$ on the element $\theta\in\Theta$. Call $H^r$ and $H^l$ respectively a right and left Haar measures of $\G$ and assume that 
\begin{enumerate}
\item $\G$, $\mathcal X$ and $\Theta$ are isomorphic.
\item The prior measure $\Pi^r$ is the measure induced by $H^r$ on $\Theta$.
\item The measure induced by $H^l$ on $\mathcal X$ is absolutely continuous with respect to the Lebesgue measure. Call $\pi^l$ the corresponding density.
\item The marginal density of $x$ is finite, so that the posterior measure $\Pi^r_{x}$ on $\Theta$, classically defined by the equation (\ref{eq_posteri}), defines the posterior probability $\Pr(.|x_0)$.
\end{enumerate}
Then, the PLR defined by the equation (\ref{def_PLR}) can be reexpressed for any $\zeta>0$ as the frequentist integral:
\begin{align}
\PLR(x_0,\zeta) &= \Pr\left(~\frac{p(x_0|\theta_0)}{\pi^l(x_0)} ~ \le ~ \zeta~ \frac{p(x|\theta_0)}{\pi^l(x)} \mid \theta_0\right)  
\end{align}
where $x_0\in\mathcal X$ is the observed data and $\theta_0\in\Theta$ the parameter value under the null hypothesis.
\end{theorem}

A more general theorem (Theorem 2) derived in a frame which avoids the Lebesgue
assumption and may involve more technical conditions is proved in Appendix 2.
Theorem 1 is a consequence of Theorem 2 and its proof is given in Appendix 3.

The assumption that $\G$ and $\X$ are isomorphic is easily relaxed by replacing the sample space by the space of a sufficient statistic. Recall that if $X$ is a random variable whose probability distribution is parametrized by $\theta$, $S(X)$ is called a sufficient statistic of $\theta$ if the probability distribution of $X$ conditioned upon the random variable $S(X)$ does not depend on $\theta$. Note that according to the \citet{darmois35} theorem, among families of probability distributions whose domains do not vary with the parameter being estimated, only in exponential families is there a sufficient statistic whose dimension remains bounded as the sample size increases. 

The expression of the theorem \ref{th1} is simply extended by replacing $X$ by a sufficient statistic $S(X)$ in the assumptions and by replacing in the frequentist integral the probability density of $X$ by the one of $S(X)$: 
\begin{corollary}
\label{th1_suff}
Call $\mathcal P_\Theta = \{p(.|\theta),\theta\in\Theta\}$ a family of probability densities with respect to any measure on $\mathcal X$. Call $S(X)$, for $X\in\X$, a sufficient statistic of $\theta$ and $\mathcal P_{S,\Theta} = \{p_S(.|\theta),\theta\in\Theta\}$ the family of probability densities of $S(X)$ with respect to the Lebesgue measure on $S(\X)$. Call $\G$ a group acting on $S(\X$). Assume that $\mathcal P_{S,\Theta}$ is invariant under the action of the group $\G$ on $S(\mathcal X)$ and note $\bar g\theta$ the induced action of the element $g\in\G$ on the element $\theta\in\Theta$. Call $H^r$ and $H^l$ respectively any right  and left Haar measures of $\G$. 
Assume that 
\begin{enumerate}
\item $\G$, $S(\mathcal X)$ and $\Theta$ are isomorphic.
\item The prior measure $\Pi^r$ is the measure induced by $H^r$ on $\Theta$.
\item The measure induced by $H^l$ on $S(\mathcal X)$ is absolutely continuous with respect to the Lebesgue measure. Call $\pi^l$ the corresponding density.
\item The marginal density of $x$ is finite, so that the posterior measure $\Pi^r_{x}$ on $\Theta$, classically defined by the equation (\ref{eq_posteri}), defines the posterior probability $\Pr(.|x_0)$.
\end{enumerate}
Then, the PLR defined by the equation (\ref{def_PLR}) can be reexpressed, with $x_0\in\mathcal X$, $\theta_0\in\Theta$ and $\zeta>0$, as the frequentist integral:
\begin{align}
\PLR(x_0,\zeta) &= \Pr\left(~\frac{p_S(S(x_0)|\theta_0)}{\pi^l(S(x_0))} ~ \le ~ \zeta~ \frac{p_S(S(x)|\theta_0)}{\pi^l(S(x))} \mid \theta_0\right)  
\end{align}
where $x_0\in\mathcal X$ is the observed data and $\theta_0\in\Theta$ the parameter value under the null hypothesis.
\end{corollary}
The proof follows the proof of the theorem \ref{th_lebesgue} in the Appendix 3.

By evaluating $\zeta=1$ in the result, the PLR with $\zeta=1$ is easily and finally shown to be equal to a frequentist p-value, where the test statistic is a weighted marginal likelihood of the sufficient statistic $S(x)$.
\begin{corollary}
\label{cor_pval}
Under the assumptions of the corollary \ref{th1_suff}, the PLR with inner threshold $\zeta=1$ is equal to a p-value:
\begin{equation}
\PLR(x_0,1) = \pval\big(T(x_0)\big)
\end{equation}
with the test statistic 
\begin{equation}
T(x) = \frac{p_S(S(x)|\theta_0)}{\pi^l(S(x))} \label{statInvar} 
\end{equation}
\end{corollary}
The corollary \ref{cor_pval} can be reexpressed as the fact that under the invariance assumptions, rejecting $\Hzero$ when $\PLR(x_0,1) > p$ is equivalent to rejecting $\Hzero$ when $\pval\big(T(x_0)\big) > p$ where the p-value is based on the idea of rejecting $\Hzero$ when $T(x_0)$ defined in equation (\ref{statInvar})
(observed weigthed likelihood under $\Hzero$) is not large enough.

\subsection{Examples and perspective}
\label{sec_ex_pers}

\citet{dempster74} has shown that the PLR is equal to the classical p-value associated to the test statistic $T(x) = |\bar x-\theta_0|$ when testing the mean of a normal family for $X$ with a uniform prior on $\Theta$. The corollary \ref{cor_pval} extends this result since the normal family is one of the distributions invariant under translation when testing the location parameter, the uniform prior (i.e. Lebesgue measure) is the measure induced from the right Haar measure associated to translation, and the test statistic $T(.)$ is a monotone function of $p_S(S(.)|\theta_0)\pi^l(S(.))^{-1}$ since the translation (sum) is commutative, so that $\Delta(g)=1$ for all $g\in\G$ and so $\pi^l$ is constant. 

The result proved here concerns all distributions invariant under some group transformation, under the assumptions that there exists a sufficient statistic and that the sets $\G$, $S(\X)$ and $\Theta$ are isomorphic. Assume for example that the likelihood $p_S$ has the typical form $p_S(S(x)|\theta) = \theta^{-1}f\big(S(x)\theta^{-1}\big)$. The likelihood is invariant under the scale transformation $g(S(x)) = \alpha \times S(x)$ and the actions on $S(\X)$ and $\Theta$ are identical. Note that $Uf(U)$ with $U = S(X)\theta^{-1}$ is a pivotal quantity, meaning that its distribution does not depend on $\theta$. The induced prior measure is classically given by $\Pi^r(d\theta) \propto \theta^{-1}d\theta$. Since the multiplication transformation is commutative, the modulus $\Delta$ is uniformly equal to 1, so that the test statistic that appears in the p-value (corollary \ref{cor_pval}) is simply $T(x) = S(x)\theta_0^{-1}f\big(S(x)\theta_0^{-1}\big)$ where $\theta_0$ is the value of the parameter under $\Hzero$. For a more general insight into the relationship between Haar invariance and the Fisher pivotal theory, see \citet{eaton99}. 

The theorem \ref{th1} assumes that $\G$, $\X$ and $\Theta$ are isomorphic. This assumption is relaxed in the corollaries \ref{th1_suff} and \ref{cor_pval} where the sample $X$ is replaced by a sufficient statistic $S(X)$: $\G$, $S(\X)$ and $\Theta$ are assumed to be isomorphic. This trick is one of the two classical dimensionality reduction techniques concerning Haar measures applied to statistical problems and somehow restricts the likelihood to belong to the exponential family from Darmois theorem. The second trick consists schematically in replacing $S(\X)$ by the orbit of $\G$ associated to the observed dataset $O_{x_0} = \{gx_0 \mid g\in\G\} \subset\X$. However, the whole set of assumptions that would be involved is more technical, see for example the general assumptions made by \citet{zidek69} or \citet{eaton02}, and not investigated here. 

\subsection{Connection to other Bayesian and frequentist reconcilations}
\label{sec_connect}

The result, which concerns hypothesis testing, may be related to the different approaches used to reconcile frequentist and Bayesian point estimation somehow and confidence interval especially. 

Group invariance applied to invariant inference is the classical frame of such unifications. The Fisherian pivotal theory (\citet{fisher73}) is an important contribution mainly to the ``frequentist'' side and the right Haar measure to the ``Bayesian'' side. The reconciliation of the two approaches has started with \citet{fraser61} and has been deeply studied since then, by \citet{zidek69} for example. The most general stage of unification is reached by \citet{eaton99}. They explicit the central hypothesis of the Fisherian pivotal theory and show under quite standard assumptions in invariance that this hypothesis leads to a procedure which is identical to the Bayesian invariant procedure when using the prior induced by the right Haar measure. Note that they also show (and in a more general manner by \citet{eaton02}) that for a Bayesian invariant inference to be admissible (in the sense that there exists no invariant inference whose mean quadratic error is lower for all $\theta$) it has to be obtained from the right Haar prior.

More concretely, the question related to reconciled probability domains is: ``Under what assumptions does the following equality hold?''
\begin{align}
\Pr\big(\theta\in \mathcal R(x)\big|x\big) &= \Pr\big(\theta\in \mathcal R(x)\big|\theta\big)  \\
\text{i.e. }~ \int_{\{\theta\in \mathcal R(x)\}} \hspace{-1.cm} d\theta ~ \pi(\theta|x) &= \int_{\{x \mid \theta\in \mathcal R(x)\}} \hspace{-1.4cm}dx ~ p(x|\theta)  \nonumber
\end{align}
For the equality to hold, each probability needs to be a constant. After \citet{fraser61} initial work, \citet{stein65} sketched the first conditions of what would be called later Stein's theorem for invariant domains. The part which is common to the different ``Stein's theorems'' is the following:
\begin{quote}{\it
 If a domain $\mathcal R(x)\subset\Theta$ satisfies $\bar g\mathcal R(x) = \mathcal R\big(g(x)\big)$
with $\bar g\mathcal R(x) = \{\bar g\theta \mid \theta\in \mathcal R(x)\}$, then under [\textit{some invariance assumptions}],
\begin{align}
\Pr\big(\theta\in \mathcal R(x)\big|x\big) &= c ~~ \forall x \in \mathcal X ~~ \text{(Bayesian probability)} \nonumber \\
\text{and } \Pr\big(\theta\in \mathcal R(x)\big|\theta\big) &= c ~~ \forall \theta \in \Theta ~~ \text{(frequentist probability)} \nonumber
\end{align}}  
\end{quote}
One of the simplest set of assumptions found since \citet{stein65} is the one of \citet{chang86}. It is relatively close to the one used for our results, presented in the section \ref{sec_equi}.

Our result, mainly holding in the theorem \ref{th_lebesgue}, is not a consequence of Stein's theorem because the domain $\mathcal R(x)\subset\Theta$ is not invariant in our case. $\mathcal R(x)$ would be invariant only if $\theta_0$ was invariant under the transformations group $\mathcal G$, i.e. if $\bar g\theta_0 = \theta_0$ for all $\bar g$ (this is equivalent to assuming that $H_0$ is invariant under $\mathcal G$). But in the theorem \ref{th1}, expressed and proved in the appendix 2 and used in the appendix 3 to prove the theorem \ref{th_lebesgue}, $\phi_{\theta}$ is assumed to be one-to-one for all $\theta\in\Theta$, which implies that $\bar g\theta_0 = \theta_0$ is equivalent to $\bar g=e$ (identity function). So the domain $\mathcal R(x)\subset\Theta$ is not invariant in our case and Stein's theorem does not imply the reconciliation result presented in the section \ref{sec_equi}. 



The theorem \ref{th_lebesgue} does not answer the previous question, but rather relaxes the form of the domain and accepts a procedure that varies according to the observed dataset $x_0$ and the value of the parameter $\theta_0$ under $\Hzero$. It answers to the question: ``Under what assumptions and for what domains $\mathcal R$ and $\mathcal C$ does the following equality hold?''
\begin{align}
 \int_{\mathcal R(x_0,\theta_0)\subset\Theta} d\theta ~ \pi(\theta|x_0) &= \int_{\mathcal C(x_0,\theta_0)\subset\X} dx ~ p(x|\theta_0) 
\end{align}
The domains found take the form
\begin{align*}
\mathcal R(x_0,\theta_0) &= \{\theta \mid p(x_0|\theta_0) \le p(x_0|\theta)\} \\
\mathcal C(x_0,\theta_0) &= \{x \mid p(x_0|\theta_0)f(x_0) \le p(x|\theta_0)f(x)\}
\end{align*}
where $f(x)$ is some weighting function, actually given by the inverse of the left prior induced by the underlying group.

\section{PLR for composite vs composite hypotheses testing}
\label{sec_generalization}

Up to this section, the PLR has been only defined in the simple ($\Hzero: \theta=\theta_0$) vs composite case, ie according to \citet{dempster74}'s first definition. 

For the more general hypothesis $\Hzero: \Psi(\theta)=\psi_0$ presented at the end of the section \ref{sec_class_test}, Dempster's approach has been generalized by \citet{aitkin97}, with a modification presented by \citet{aitkin10}. Namely, \citet{aitkin10} proposes to compute 
$\Pr\left(p(x|\theta) < p\left(x|(\Psi,\Lambda)^{-1}(\psi_0,\Lambda(\theta))\right) \big| x\right)$
and details and illustrates some advantages of the method. In the case of $\Psi(\theta)=\theta$, it corresponds to Dempster's definition (see page 42 of \citet{aitkin10}). The approach of \citet{evans97} also carries interesting properties. In particular, a variety of optimality properties for inferences based on relative belief ratios are established in \citet{evans06}, \citet{evans08} and \citet{evans11}, which include optimal testing properties based on establishing a kind of Bayesian version of the Neyman-Pearson lemma. 

However, the hypotheses test case on which they rely is not broad enough for many cases. The purpose of this section is to extend the definition of the PLR to the classical composite vs composite hypotheses test. 

Suppose the data models related to the two hypotheses belong to the same parametric family $\mathcal P_{\Theta} = \{p(.|\theta),\theta\in\Theta\}$. This assumption can actually be realized for any hypotheses test of parametric models by merging the tested parametric families in a so-called \textit{super-model}. A composite vs composite hypotheses test consists in choosing among
\begin{align}
\Hzero : \theta\in\Theta_0 ~~~~~ \Hone : \theta\in\Theta_1
\end{align}
for any domains $\Theta_0$ and $\Theta_1$. We note $\Pi_0(.)$ and $\Pi_1(.)$ the prior distributions over $\Theta_0$ and $\Theta_1$.

In this section we propose two extensions of Dempster's approach for this test case. The first extension proposed can be used when the prior under one hypothesis is improper but both posteriors are proper. The second extension, made of two symmetrical probabilities, is the statistics suggested by a new Bayesian-type Neyman-Pearson lemma which also indicates that the LR is a central discrepancy variable.

\subsection{Extensions of the PLR}
\label{sec_ext_plr}

In the simple $\Theta_0=\{\theta_0\}$ vs composite hypotheses test, the PLR was primarly defined as 
\begin{equation}
\PLR(x,\zeta) = \int_{\{\theta_1 \mid p(x|\theta_1) < \zeta p(x|\theta_0)\}} \Pi_1(d\theta_1|x)\nonumber
\end{equation}

In the composite vs composite hypotheses test, a first interesting extension of this concept consists in defining  the following statistics:
\begin{align}
\PLR_{01}(x,\zeta) &= \int_{\{(\theta_0,\theta_1) \mid p(x|\theta_0) < \zeta p(x|\theta_1)\}} \Pi_0(d\theta_1|x)\Pi_1(d\theta_0|x) \label{eq_plr01}
\end{align}
It is well defined as soon as the posterior distributions are both proper. Since only $x$ is known, the event $p(x|\theta_0) < \zeta p(x|\theta_1)$ can be measured only by integrating over all $\theta_0\in\Theta_0$ and all $\theta_1\Theta_1$. Here we decide to measure it according to the posterior distribution of $\theta_0$ times the posterior distribution of $\theta_1$, which is perfectly allowed.

A second interesting extension of the simple PLR consists in defining  the two symmetrical following statistics:
\begin{align}
\PLR_0(x,\zeta) &= \int_{\{(\theta_0,\theta_1) \mid p(x|\theta_1) < \zeta p(x|\theta_0)\}} \Pi_0(d\theta_0|x)\Pi_1(d\theta_1) \label{eq_plr0}\\
\PLR_1(x,\zeta) &= \int_{\{(\theta_0,\theta_1) \mid p(x|\theta_0) < \zeta p(x|\theta_1)\}} \Pi_1(d\theta_1|x)\Pi_0(d\theta_0) \label{eq_plr1}
\end{align}

In the simple vs composite test, note that only $\PLR_{01}(x,\zeta)$ and $\PLR_1(x,\zeta)$ are equal to the PLR as defined by \citet{dempster74} and can thus be considered as extensions of the PLR. However, given the symmetry of the two hypotheses in a composite vs composite test, the notation $\PLR_0(x,\zeta)$ will be also necessary in the sequel.

Each quantity has its own definition, interpretation, properties and field of use. We don't investigate interpretation far here, and rather focus on unquestionable properties and results.

$\PLR_{01}(x,\zeta)$ is the only extension of the two which allows for using improper priors. It will be illustrated in the next subsection to test a practical precipitation change, which requires the use of a prior which is too smooth for the other extension to be used. 

On the other side, the statistics $\PLR_1(x,1)$ is the expectation over the prior under $H_0$ of the posterior probability under $H_1$ that the likelihood of $\theta_0$ is less than the likelihood of $\theta_1$, and reciprocally.
\begin{equation}
\PLR_1(x,\zeta) = \esp_0[\Pr_1(p(x|\theta_0) < \zeta p(x|\theta_1) | x)]  \nonumber
\end{equation}
$\PLR_0$ and $\PLR_1$ will appear as statistics emerging from a more general frame through a Bayesian-type Neyman-Pearson lemma.

Extending the interpretation of the new PLRs in terms of joint probabilities requires the definition of a measure over $\Theta_0\times\Theta_1$ given $x$ and one of the two hypotheses. Such a measure seems to make sense in terms of both mathematics and interpretation but the issue needs to be deepened. 
\begin{remark}
\label{lem_joint_meas}
If all subsets defined on the sets $\Theta_0\times\X | H_0$ and $\Theta_1 | H_0$ are independent, then the joint measure $\Pi_{01,0}$ defined over $\Theta_0\times\Theta_1\times\X | H_0$ is equal to:
\begin{equation} 
\Pi_{01,0}(d\theta_0,d\theta_1 | x) = \Pi_0(d\theta_0 | x) \Pi_1(d\theta_1)  \nonumber
\end{equation}
 for infinitesimal subsets around any $(\theta_0,\theta_1) \in \Theta_0\times\Theta_1$. The same holds when replacing the roles of $H_0$ and $H_1$, and leads to the measure $\Pi_{01,1}$:
\begin{equation}
\Pi_{01,1}(d\theta_0,d\theta_1 | x) = \Pi_0(d\theta_0) \Pi_1(d\theta_1 | x)  \nonumber
\end{equation}
\end{remark}

The proof of the remark stands in the appendix 4. So if we assume that the joint measures exist and that the priors and posteriors are all proper, then the composite PLRs defined in the equations (\ref{eq_plr0}) and (\ref{eq_plr1}) are probability measures.



\subsection{Example: detection of a change in precipitation in Switzerland}
\label{sec_examp}

Let's illustrate $\PLR_{01}$ defined in the equation (\ref{eq_plr01}). 

Although the change in temperature in the 20th century is evident at a world scale and in some areas, a potential change in precipitation remains under study. As a simple case, let's consider a single weather station in Switzerland and test whether the statistical properties of the rain frequency have changed. 

As recalled for example by \citet{aksoy00}, daily precipitation amounts are well described by a gamma distribution, characterized by a shape parameter $a$ and a rate parameter $b$. Assume the daily rainfalls $x_1$ fallen during the five first automns of the 20th century are i.i.d. with parameters $a_1$ and $b_1$, as well as $x_2$ during the five last automns with parameters $a_2$ and $b_2$. The detection of a statistical change consists in testing whether the set of parameters are equal or not:
\begin{align}
\Hzero : (a_1,b_1) = (a_2,b_2) ~~~~~ \Hone : (a_1,b_1) \neq (a_2,b_2)
\end{align}

Note that the dimension of $\Theta_0 = \mathbb R_{+*}^2$ is less than the dimension of $\Theta_1 = \mathbb R_{+*}^4$, so that for a regular prior under $\Hone$, $\Pr_1(\theta\in\Theta_0)=0$. \citet{borges07} are particularly interested by the behavior of the e-value of the FBST in such cases. Here it simply means that there is one prior $\pi(a,b)$ under $\Hzero$ and the product of two priors $\pi(a_1,b_1)\times\pi(a_2,b_2)$ under $\Hone$, to be combined respectively with the likelihood $p(x_1,x_2|a,b)$ under $\Hzero$ and the likelihood $p(x_1|a_1,b_1)\times p(x_2|a_2,b_2)$ under $\Hone$. 

To enable simple simulations of the posterior distributions under both hypotheses, the conjugate prior (see the compedium by \citet{fink97}) of the gamma distribution developed by \citet{miller80} is used for $\pi$, with hyperparameters that may vary without affecting much the final results. The impact of the prior on the PLR is easy to see from the PLR display as will be explained very shortly. In practice, the prior $\pi$ is almost improper so that only the $\PLR_{01}$ defined in equation (\ref{eq_plr01}) can be used. 

First, simulations roughly corresponding to the observed rainfall are performed. One dataset is simulated under $H_0$ and another is simulated under some reasonably similar alternative $H_1$. The two simulated datasets are characterized by their likelihoods, displayed on the figure \ref{fig_like}.

\begin{figure}
\begin{center}
\includegraphics[width=0.33\columnwidth]{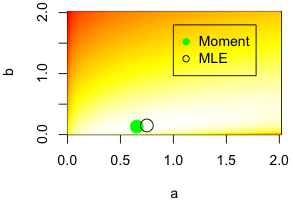}\includegraphics[width=0.33\columnwidth]{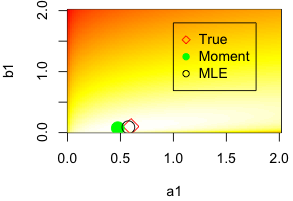}\includegraphics[width=0.33\columnwidth]{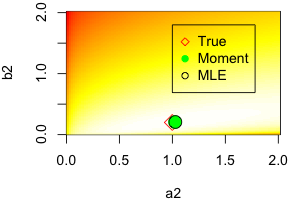}
\caption{Log-likelihood of the dataset under $\Hzero$ (left figure) and the dataset under $\Hone$(center and right figures). Some frequentist estimations of the parameters (circles) are superimposed on the true parameters values (diamond).\label{fig_like}}
\end{center}
\end{figure}

The posterior distribution of each couple $(a,b)$, $(a_1,b_1)$ and $(a_2,b_2)$ is separately sampled by a MCMC multivariate slice sampling algorithm (\citet{radford03}) implemented in the R package ``SamplerCompare'' kindly written and provided by \citet{thompson12}. The PLR is simply computed by ordering the LR obtained for all possible combinations of parameters and counting the fraction which is less than some threshold $\zeta$ chosen according to the level of evidence wanted in favor of $\Hzero$ or $\Hone$. In practice, the PLR is displayed as a function of $\zeta$ by simply displaying the empirical cumulative distribution of the LRs. This leads to the figure \ref{fig_plr_simu}. It can be read for example that for the dataset under $\Hzero$, $\PLR_{01}(x,0.1) = 0.08$, which means that there is an almost null probability that the likelihood under $\Hone$ is more than 10 times greater than the likelihood under $\Hzero$, so that $\Hzero$ is (correctly) clearly accepted. Alternatively, for the $\Hone$ dataset, $\PLR_{01}(x,0.1) = 1.00$, meaning that there is a probabity one that the likelihood under $\Hone$ is more than 10 times greater than the likelihood under $\Hzero$, so that $\Hzero$ is (correctly) clearly rejected. 

Note that since the GLR indicates the lower bound of the support of the PLR and since the slope of the PLR is infinite there if the likelihood function is smooth enough at its maximum (see section \ref{sec_intro_plr}), the prior exact expression only affects the way $\PLR_{01}(x,\zeta)$ increases as $\zeta$ departs from $\GLR$. Here for example the choice of the hyperparameters (among a domain considered as \textit{reasonable}) does not change the conclusion that would be drawn from the PLR displayed on the figure \ref{fig_plr_simu}.

\begin{figure}
\begin{center}
\includegraphics[width=0.4\columnwidth]{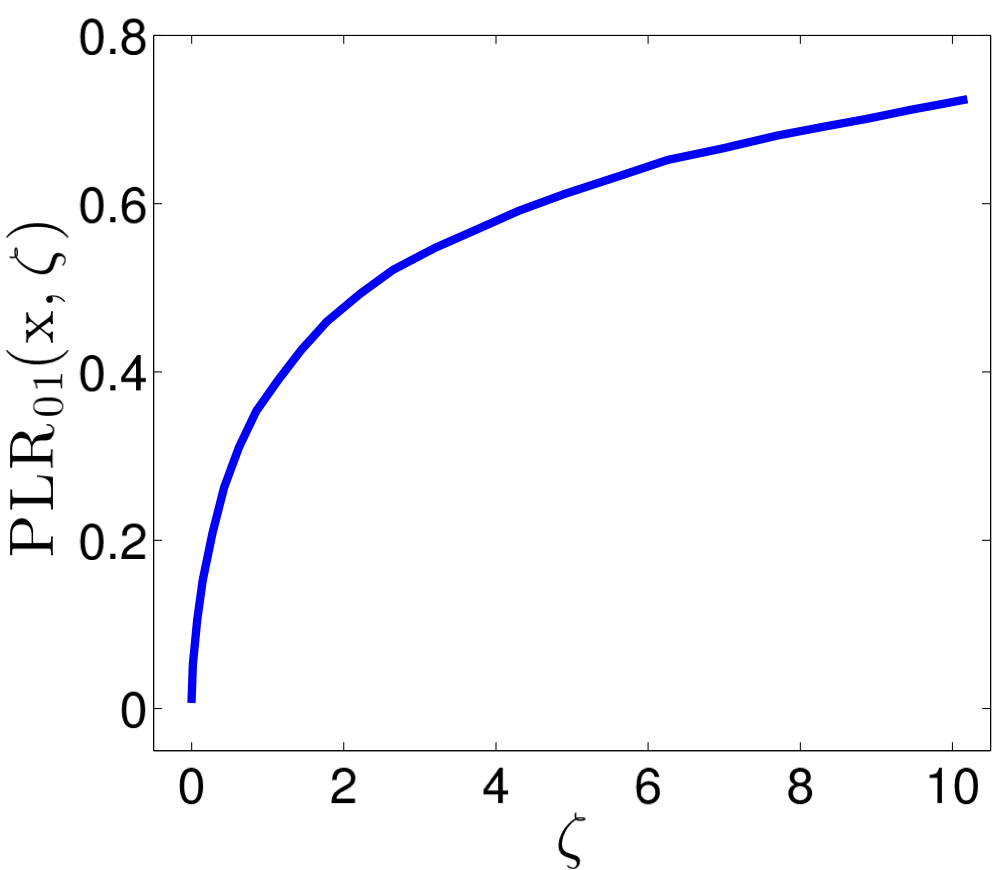}\quad \includegraphics[width=0.4\columnwidth]{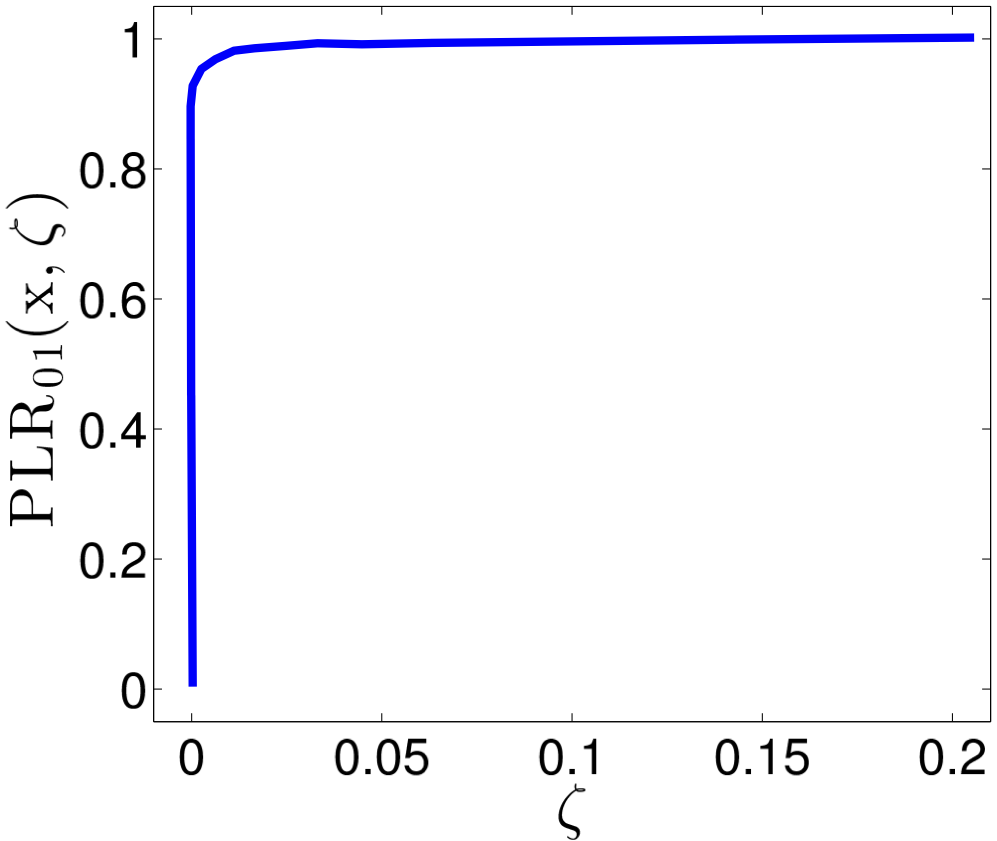}

\caption{PLR obtained from the dataset simulated under $\Hzero$ (left) and the dataset simulated under $\Hone$ (right). In practice those are simply the empirical cumulative distributions of the $\LR(x,\theta^{[i]})$ chains. For the $\Hzero$ dataset the PLR clearly correctly accepts $\Hzero$, and reciprocally for the $\Hone$ dataset the PLR clearly correctly rejects $\Hzero$: notice the difference of the x axes scales between both simulation cases. \label{fig_plr_simu}}
\end{center}
\end{figure}

Switching to the true dataset $x$, the PLR is obtained following the same procedure as with the simulated datasets. $\PLR_{01}(x,\zeta)$ is displayed on the figure \ref{fig_plr_obs}. The graph is --by construction of the simulations-- very similar the one obtained for the graph obtained with the data simulated under $\Hzero$. Now, $\PLR_{01}(x,0.1) = 0.10$ and $\Hzero$ can clearly not be rejected, so that no change in the 20th precipitation in Switzerland is detected, which is not surprising to climatologists.

\begin{figure}
\begin{center}
\includegraphics[width=0.5\columnwidth]{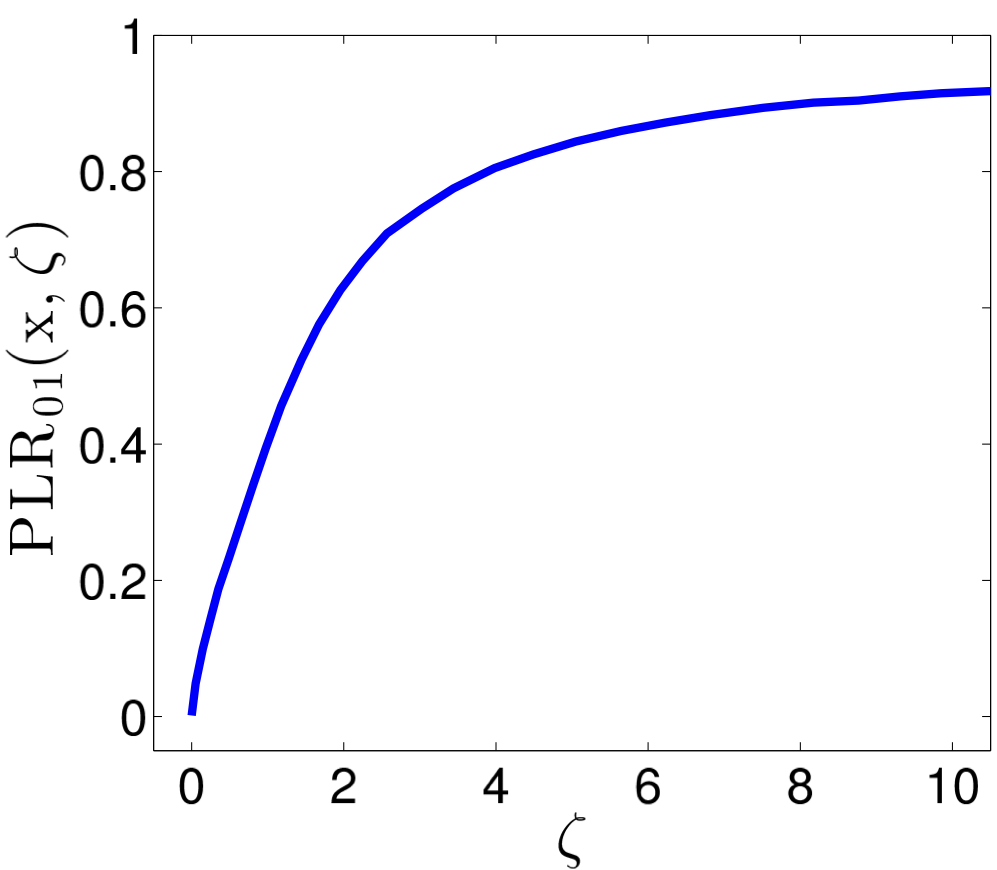}
\caption{PLR obtained from daily automnal precipitation observed in a weather station in Switzerland from 1900-1905 for the first part of the dataset and 1995-2000 for the second part of the dataset. $\Hzero$ cannot be rejected, so that no change in precipitation is detected.\label{fig_plr_obs}}
\end{center}
\end{figure}

\subsection{Bayesian type Neyman-Pearson lemma}
\label{sec_bayes_neyman}

In the choice of an hypothesis, instead of considering the subset 
\begin{equation}
\mathcal{R}^{\ast}(x) = \{(\theta_0,\theta_1) \mid p(x|\theta_0) < \zeta p(x|\theta_1)\}  \label{eq_set_lr}
\end{equation}
one might consider any subset $\mathcal{R}(x)\subset\Theta_0\times\Theta_1$, that may depend on $x$. Such a subset could involve a discrepancy variable $D : \X\times\Theta \mapsto \R$ like in the predictive p-value highlighted by \citet{meng94}, and take the form ``$D(x,\theta_0) < \zeta D(x,\theta_1)$''. The discrepancy variable that appears in the PLR is $\LR(x,\theta)$.

$\mathcal{R}^{\ast}(x)$ defined from the LR test is interesting for hypotheses testing because this set is a somehow classical \textit{hypothesis} rejection set. It is not a fully classical rejection set because it is defined on the parameter space rather than on the observation space, but its characterization is optimal in the frequentist setting. $\mathcal{R}^{\ast}(x)$ is the set, depending on the dataset $x$, of all fixed $(\theta_0,\theta_1) \in \Theta_0\times\Theta_1$ such that the likelihood of $\theta_0$ is less than the likelihood of $\theta_1$, which reasonably leads to reject $H_0$ for this element $(\theta_0,\theta_1)$. The same way, one can replace the LR test by any test, ie consider any subset $\mathcal{R}(x)\subset\Theta_0\times\Theta_1$ such that for $(\theta_0,\theta_1)\in\mathcal{R}(x)$, $H_0$ would be decided to be rejected. 

With such a phrasing, it may appear natural that the frequentist Neyman-Pearson lemma can be derived in a reciprocal, somehow Bayesian, frame. Note that the Neyman-Pearson lemma can be expressed, as will be the proposition here, symmetrically in the two hypotheses. The symmetry is only broken when adopting the Neyman paradigm which fixes a level for the PFA and deduce the corresponding $\zeta$ (see section \ref{sec_class_test}).

To rederive a Neyman-Pearson lemma one would define the reciprocal notions of ``Probability of False Alarm'' and ``Probability of good Detection'':
\begin{align}
\PFA_B(\mathcal{R},x) &= \int_{\mathcal{R}(x)} \Pi_0(d\theta_0|x)\Pi_1(d\theta_1) \label{eq_pfab} \\
\PD_B(\mathcal{R},x) &= \int_{\mathcal{R}(x)} \Pi_1(d\theta_1|x)\Pi_0(d\theta_0) \label{eq_pdb}
\end{align} 
These quantities would also define probability measures if the joint measures exist and if the priors and posteriors are all proper.

Note that these measures can be related to a joint measure with no conditioning over the hypothesis: for any set $\mathcal{R}(x) \subset \Theta_0\times\Theta_1$ eventually depending on $x$,
\begin{align}
\Pr(\mathcal{R},x) &= \Pr(H_0) ~\Pr(\mathcal{R} | x,H_0) +  \Pr(H_1)~\Pr(\mathcal{R} | x,H_1)\nonumber \\
\text{with }~ \Pr(\mathcal{R} | x,H_i) &= \int_\mathcal{R} \Pi_{01,i}(d\theta_0,d\theta_1|H_i,x) \nonumber\\
  &= \int_\mathcal{R} \Pi_i(d\theta_i|x) \Pi_j(d\theta_j) \nonumber \\
\text{so }~  \Pr(\mathcal{R} | x) &= \Pr(H_0)  \int_\mathcal{R} \Pi_0(d\theta_0|x)\Pi_1(d\theta_1) +  \Pr(H_1) \int_\mathcal{R} \Pi_1(d\theta_1|x)\Pi_0(d\theta_0) \nonumber\\
&= \Pr(H_0) ~\PFA_0(\mathcal{R},x) +  \Pr(H_1) ~\PD_1(\mathcal{R},x) \nonumber
\end{align}
The Bayesian type probabilities $\PFA_B$ and $\PD_B$ add up the same way type I and type II probability errors add up in a frequentist integral.  Note also that $\PFA_B(\bar{\mathcal{R}},x) = 1 - \PFA_B(\mathcal{R},x)$ where $\bar{\mathcal{R}}(x)$ is the set complementary to $\mathcal{R}(x)$ in $\Theta_0\times\Theta_1$.


Following the underlying idea of the Neyman-Pearson approach, a possibility for choosing $\mathcal{R}(x)$ consists in maximizing $\PD_B(\mathcal{R},x)$ over $\mathcal{R}(x)$ for a fixed $\PFA_B(\mathcal{R},x)$. 
\begin{proposition} \label{lemma_np_b}
The subset that maximizes $\PD_B(\mathcal{R},x)$ for a fixed value of $\PFA_B(\mathcal{R},x)$ is equal to the LR subset $\mathcal{R}^{\ast}(x)$ defined in equation (\ref{eq_set_lr}). In this case, the ``Bayesian PFA and PD'' are given by $\PFA_B(\mathcal{R}^{\ast},x) = 1-\PLR_0(x,\zeta)$ and $\PD_B(\mathcal{R}^{\ast},x) = \PLR_1(x,\zeta)$.\\
Reciprocally, the subset that maximizes $\PFA_B(\mathcal{R},x)$ for a fixed value of $\PD_B(\mathcal{R},x)$ is equal to $\bar{\mathcal{R}^{\ast}}(x)$, ie the set which accepts $H_0$ according to the LR test. In this case, $\PFA_B(\bar{\mathcal{R}}^{\ast},x) = \PLR_0(x,\zeta)$  and $\PD_B(\bar{\mathcal{R}}^{\ast},x) = 1-\PLR_1(x,\zeta)$. 
\end{proposition}
As postdata measures (i.e. depending on the observed data), contrary to the predata frequentist PFA and PD, it is therefore informative enough to give $\PLR_0(x,\zeta)$ and $\PLR_1(x,\zeta)$ for some value $\zeta$ of interest. But this is only possible if the priors and posteriors under both hypotheses are proper.

The proof of the proposition follows the proof of the Neyman-Pearson lemma restricted to deterministic tests. It stands in the appendix 5.

\section{Concluding general discussion about the PLR}
\label{sec_conclu}

The PLR introduced by \citet{dempster74} in the simple vs composite hypotheses test deserves much attention. It compares the original likelihoods $p(x|\theta_0)$ and $p(x|\theta_1)$ by computing the posterior probability that this usual LR test chooses $\Hzero$ or $\Hone$. The PLR is simple, nicely interpretable and coupled with some deep properties. Compared to the classical Bayesian hypotheses tests, first note that unlike the BF, the PLR can be defined even for improper priors, and unlike $\Pr(\Hzero|x)$ it does not require the delicate choice of some $\Pr(\Hzero)$. This is crucial in practice as well as in fundamental issues like Lindley's paradox.

The PLR also turns out to be a very natural alternative to the BF in many aspects. The PLR first compares (the original likelihoods) and then integrates, whereas the BF first integrates and then compares (the marginal likelihoods). In the simple vs composite hypotheses test, considering $\LR(x,\theta)$ as a random variable for a fixed $x$, the PLR is its posterior cumulative distribution (i.e. the probability of a one sided \textit{credible interval}) whereas the BF is its posterior mean \textit{point estimate}. This credible interval vs point estimate duality between the PLR and the BF also translates in decision theory: \citet{hwang92} stressed that $\Pr(\Hzero|x)$ does not measure evidence, since this is done only through the likelihood, but measures the accuracy of a test by \textit{estimating} the indicator function $I_{\Theta_0}(\theta)$. Also note that being the measure of a credible interval, the PLR is also a natural hypotheses test tool which connects postdata (i.e. conditioned upon $x$) hypotheses testing and credible interval inference. This formal equivalence was known to hold for predata inference (a rejection set is equivalent to a confidence interval) and ``known'' not to hold for postdata inference for usual Bayesian tools (see \citet{lehmann05} and \citet{goutis97}). Tools like the PLR set up this connection.

However, when generalizing the PLR in the section \ref{sec_ext_plr}, most of these dual properties cannot be generalized to the composite vs composite hypotheses test. Instead, a reciprocity between the PLR and the BF exists through a Neyman-Pearson lemma perspective. The second extension of the PLR has been shown in the section \ref{sec_bayes_neyman} to be a somehow optimal measure, in that it measures the set that maximizes $\PD_B$ for a fixed $\PFA_B$ (Bayesian-type version of the frequentist Neyman-Pearson lemma). Reciprocally, the BF gives a somehow optimal measure, although in the \textit{frequentist} Neyman-Pearson sense, in that it maximizes the average over $\pi_1$ of $\PD(\theta_1)$ for a fixed $\PFA$ (frequentist classical Neyman-Pearson lemma but for the marginal likelihood and not the original unknown one). 

In the simple vs composite hypotheses test, the connection between the PLR (related to credible interval) and the BF (related to point estimate) has been underlined. Another important connection lies between frequentist and Bayesian type hypotheses tests, namely frequentist p-values and $\Pr(\Hzero|x)$ or PLR. This reconciliation quest has been the subject of many debates, including Lindley's paradox in its most simple form (test of the mean of a Gaussian with a uniform prior), which has only been simply reached by the PLR by \citet{dempster74}. In the section \ref{sec_equi} we have generalized this reconciliation result to a quite general invariant frame, close to the one used in Stein's theorem, i.e. in a frame under which confidence and credible intervals are equivalent. Note that invariance is also a perspective adopted to develop and evaluate inferences, and in particular to develop new p-values as done recently by \citet{evans10} for example. For the PLR, standard simple invariance properties directly follows from the simple use of the likelihoods.

To conclude on the contribution of this paper, the equivalence between the PLR and a p-value has been proved in a general invariant frame, which nicely connects to the equivalence between confidence and credible domains. This result may contribute to a better understanding of deep and fundamental issues related to both hypotheses testing and parameter estimation, in both frequentist and Bayesian paradigms.






\bibliographystyle{apalike}
\bibliography{biblio}

\begin{thebibliography}{}

\bibitem[Aitkin, 1997]{aitkin97}
Aitkin, M. (1997).
\newblock The calibration of p-values, posterior {B}ayes factors and the {AIC}
  from the posterior distribution of the likelihood.
\newblock {\em Statistics and Computing}, 7:253--261.

\bibitem[Aitkin, 2010]{aitkin10}
Aitkin, M. (2010).
\newblock {\em Statistical inference: an integrated {B}ayesian / likelihood
  approach}.
\newblock Chapman and Hall.

\bibitem[Aitkin et~al., 2005]{aitkin05}
Aitkin, M., Boys, R.~J., and Chadwick, T. (2005).
\newblock Bayesian point null hypothesis testing via the posterior likelihood
  ratio.
\newblock {\em Statistics and Computing}, 25(3):217--230.

\bibitem[Aitkin et~al., 2009]{aitkin09}
Aitkin, M., Liu, C.~C., and Chadwick, T. (2009).
\newblock Bayesian model comparison and model averaging for small-area
  estimation.
\newblock {\em Annals of Applied Statistics}, 3(1):199--221.

\bibitem[Aksoy, 2000]{aksoy00}
Aksoy, H. (2000).
\newblock Use of gamma distribution in hydrological analysis.
\newblock {\em Turk. J. Engin. Environ. Sci.}, 24:419--428.

\bibitem[Baskurt and Evans, 2013]{baskurt13}
Baskurt, Z. and Evans, M. (2013).
\newblock Hypothesis assessment and inequalities for {B}ayes factors and
  relative belief ratios.
\newblock {\em Bayesian Analysis}, 8,3:569--590.

\bibitem[Berger and Sellke, 1987]{berger87b}
Berger, J. and Sellke, T. (1987).
\newblock Testing a point null hypothesis: the irreconcilability of {P} values
  and evidence (with discussion).
\newblock {\em Journal of the American Statistical Association}, 82:112--139.

\bibitem[Berger, 1985]{berger85}
Berger, J.~O. (1985).
\newblock {\em Statistical decision theory and {B}ayesian analysis}.
\newblock Springer-Verlag, 2nd edition.

\bibitem[Berger et~al., 1994]{berger94}
Berger, J.~O., Brown, L., and Wolpert, R. (1994).
\newblock A unified conditional frequentist and {B}ayesian test for fixed and
  sequential simple hypothesis testing.
\newblock {\em Annals of Statistics}, 22(4):1787--1807.

\bibitem[Berger and Delampady, 1987]{berger87}
Berger, J.~O. and Delampady, M. (1987).
\newblock Testing precise hypotheses (with discussion).
\newblock {\em Statistical Science}, 2(3):317--335.

\bibitem[Bernardo, 2011]{bernardo10}
Bernardo, J. (2011).
\newblock {\em Bayesian Statistics 9}, chapter Integrated objective {B}ayesian
  estimation and hypothesis testing.
\newblock Oxford University Press.

\bibitem[Birnbaum, 1962]{birnbaum62}
Birnbaum, A. (1962).
\newblock On the foundation of statistical inference (with discussion).
\newblock {\em Journal of the American Statistical Association},
  57(298):269--326.

\bibitem[Borges and Stern, 2007]{borges07}
Borges, W. and Stern, J. (2007).
\newblock The rules of logic composition for the {B}ayesian epistemic e-values.
\newblock {\em Logic journal of the IGPL}, 15(5--6):401--420.

\bibitem[Casella and Berger, 1987]{casella87}
Casella, G. and Berger, R.~L. (1987).
\newblock Reconciling {B}ayesian and frequentist evidence in the one-sided
  testing problem.
\newblock {\em Journal of the American Statistical Association},
  82(397):106--111.

\bibitem[Chang and Villegas, 1986]{chang86}
Chang, T. and Villegas, C. (1986).
\newblock On a theorem of {S}tein relating {B}ayesian and classical inferences
  in group models.
\newblock {\em The Canadian Journal of Statistics}, 14(4):289--296.

\bibitem[Darmois, 1935]{darmois35}
Darmois, G. (1935).
\newblock Sur les lois de probabilit\'e \`a estimation exhaustive.
\newblock {\em Compte-Rendu de l'Acad{\'e}mie des Sciences de Paris},
  200(1265--1266).

\bibitem[Dempster, 1973]{dempster74}
Dempster, A.~P. (1973).
\newblock The direct use of likelihood for significance testing.
\newblock In {\em Proceedings of Conference on Foundational Questions in
  Statistical Inference}, pages 335--354, Aaarhus, Denmark.

\bibitem[Dempster, 1997]{dempster97}
Dempster, A.~P. (1997).
\newblock Commentary on the paper by {M}urray {A}itkin, and on discussion by
  {M}ervyn {S}tone.
\newblock {\em Statistics and Computing}, 7(4):265--269.

\bibitem[Eaton, 1989]{eaton89}
Eaton, M. (1989).
\newblock {\em Group invariance applications in Statistics}.
\newblock Regional Conf. Series in Prob. and Stat.

\bibitem[Eaton, 2007]{eaton07}
Eaton, M. (2007).
\newblock {\em Multivariate statistics}.
\newblock Institute of Mathematical Statistics.

\bibitem[Eaton and Sudderth, 1999]{eaton99}
Eaton, M. and Sudderth, W. (1999).
\newblock Consistency and strong inconsistency of group-invariant predictive
  inferences.
\newblock {\em Bernoulli}, 5(5):833--854.

\bibitem[Eaton and Sudderth, 2002]{eaton02}
Eaton, M. and Sudderth, W. (2002).
\newblock Group invariant inference and right {H}aar measure.
\newblock {\em Journal of Statistical planning and inference},
  103(1--2):87--99.

\bibitem[Evans, 1997]{evans97}
Evans, M. (1997).
\newblock Bayesian inference procedures derived via the concept of relative
  surprise.
\newblock {\em Communications in Statistics}, 26:1125--1143.

\bibitem[Evans et~al., 2006]{evans06}
Evans, M., Guttman, I., and Swartz, T. (2006).
\newblock Optimality and computations for relative surprise inferences.
\newblock {\em Canadian Journal of Statistics}, 34(1):113--129.

\bibitem[Evans and Jang, 2010]{evans10}
Evans, M. and Jang, G. (2010).
\newblock Invariant p-values for model checking.
\newblock {\em Annals of Statistics}, 38(1):512--525.

\bibitem[Evans and Jang, 2011]{evans11}
Evans, M. and Jang, G. (2011).
\newblock Inferences from prior-based loss functions.
\newblock Technical Report 1104, Dept. of Statistics, U. of Toronto.

\bibitem[Evans and Shakhatreh, 2008]{evans08}
Evans, M. and Shakhatreh, M. (2008).
\newblock Optimal properties of some {B}ayesian inferences.
\newblock {\em Electronic Journal of Statistics}, 2:1268--1280.

\bibitem[Fink, 1997]{fink97}
Fink, D. (1997).
\newblock A compendium of conjugate priors.
\newblock Technical report, Montana State University.

\bibitem[Fisher, 1956]{fisher73}
Fisher, R.~A. (1973, 1st ed.: 1956).
\newblock {\em Statistical methods and scientific inference}.
\newblock Oliver and Boyd, 3rd edition.

\bibitem[Fraser, 1961]{fraser61}
Fraser, D. A.~S. (1961).
\newblock The fiducial method and invariance.
\newblock {\em Biometrika}, 48(3--4):261--280.

\bibitem[Goutis and Casella, 1997]{goutis97}
Goutis, C. and Casella, G. (1997).
\newblock Relationships between post-data accuracy measures.
\newblock {\em Annals of the Institute of Statistical Mathematics},
  49(4):711--726.

\bibitem[Hwang et~al., 1992]{hwang92}
Hwang, J., Casella, G., Robert, C., Wells, M., and Farrell, R. (1992).
\newblock Estimation of accuracy in testing.
\newblock {\em Annals of Statistics}, 20(1):490--509.

\bibitem[Jeffreys, 1961]{jeffreys61}
Jeffreys, H. (1961).
\newblock {\em Theory of probability}.
\newblock Oxford University Press, 3rd edition.

\bibitem[Lehmann and Romano, 2005]{lehmann05}
Lehmann, E.~L. and Romano, J.~P. (2005).
\newblock {\em Testing statistical hypotheses}.
\newblock Springer, 3rd edition.

\bibitem[Lindley, 1957]{lindley57}
Lindley, D. (1957).
\newblock A statistical paradox.
\newblock {\em Biometrika}, 44(1--2):187--192.

\bibitem[Meng, 1994]{meng94}
Meng, X.-L. (1994).
\newblock Posterior predictive p-values.
\newblock {\em Annals of Statistics}, 22(3):1142--1160.

\bibitem[Miller, 1980]{miller80}
Miller, R. (1980).
\newblock Bayesian analysis of the two-parameter {G}amma distribution.
\newblock {\em Technometrics}, 22(1):65--69.

\bibitem[Nachbin, 1965]{nachbin65}
Nachbin, L. (1965).
\newblock {\em The {H}aar integral}.
\newblock Van Nostrand.

\bibitem[Newton and Raftery, 1994]{newton94}
Newton, M. and Raftery, A. (1994).
\newblock Approximate {B}ayesian inference with the weighted likelihood
  bootstrap.
\newblock {\em Journal of the Royal Statistical Society Series B}, 56(1):3--48.

\bibitem[Neyman, 1977]{neyman77}
Neyman, J. (1977).
\newblock Frequentist probability and frequentist statistics.
\newblock {\em Synthese}, 36:97--131.

\bibitem[Neyman and Pearson, 1933]{neyman33}
Neyman, J. and Pearson, E. (1933).
\newblock On the problem of the most efficient tests of statistical hypotheses.
\newblock {\em Philosophical Transactions of the Royal Society of London,
  Series A}, 231:289--337.

\bibitem[Oh and DasGupta, 1999]{oh99}
Oh, H. and DasGupta, A. (1999).
\newblock Comparison of the p-value and posterior probability.
\newblock {\em Journal of Statistical planning and inference},
  76(1--2):93--107.

\bibitem[O'Hagan, 1995]{ohagan95}
O'Hagan, A. (1995).
\newblock Fractional {B}ayes factors for model comparison.
\newblock {\em Journal of the Royal Statistical Society}, 57(1):99--138.

\bibitem[Pereira and Stern, 1999]{pereira99}
Pereira, C. and Stern, J. (1999).
\newblock Evidence and credibility: full {B}ayesian significance test for
  precise hypotheses.
\newblock {\em Entropy}, 1:104--115.

\bibitem[Radford, 2003]{radford03}
Radford, N. (2003).
\newblock Slice sampling.
\newblock {\em Annals of Statistics}, 31(3):705--767.

\bibitem[Robert, 2007]{robert07}
Robert, C.~P. (2007).
\newblock {\em The {B}ayesian choice}.
\newblock Springer, 2nd edition.

\bibitem[Robins et~al., 2000]{robins00}
Robins, J., van~der Vaart, A., and Ventura, V. (2000).
\newblock Asymptotic distribution of p-values in composite null models.
\newblock {\em Journal of the American Statistical Association},
  95(452):1143--1156.

\bibitem[Royall, 1997]{royall97}
Royall, R. (1997).
\newblock {\em Statistical evidence: a likelihood paradigm}.
\newblock Chapman and Hall / CRC Press.

\bibitem[Sellke et~al., 2001]{sellke01}
Sellke, T., Bayarri, M.~J., and Berger, J.~O. (2001).
\newblock Calibration of p-values for testing precise null hypotheses.
\newblock {\em American Statistician}, 55(1):62---71.

\bibitem[Smith, 2010]{smith10e}
Smith, I. (2010).
\newblock {\em D{\'e}tection d'une source faible : mod{\`e}les et m{\'e}thodes
  statistiques. Application {\`a} la d{\'e}tection d'exoplan{\`e}tes par
  imagerie directe.}
\newblock PhD thesis, Universit{\'e} de Nice Sophia-Antipolis.

\bibitem[Smith and Ferrari, 2010]{smith10d}
Smith, I. and Ferrari, A. (2010).
\newblock The posterior distribution of the likelihood ratio as a measure of
  evidence.
\newblock In {\em Maxent}.

\bibitem[Smith and Ferrari, 2014]{smith14}
Smith, I. and Ferrari, A. (2014).
\newblock Equivalence between the posterior distribution of the likelihood
  ratio and a p-value in an invariant frame.
\newblock {\em Bayesian Analysis}.

\bibitem[Stein, 1965]{stein65}
Stein, C. (1965).
\newblock Approximation of improper prior measures by prior probability
  measures.
\newblock In {\em Bernoulli, Bayes, Laplace Festschrift}, pages 217--240.
  Springer-Verlag.

\bibitem[Thompson, 2012]{thompson12}
Thompson, M. (2012).
\newblock {\em R package ``{S}ampler{C}ompare''}.

\bibitem[Tsao, 2006]{tsao06}
Tsao, C.~A. (2006).
\newblock A note on {L}indley's paradox.
\newblock {\em Test}, 15(1):125--139.

\bibitem[Villegas, 1981]{villegas81}
Villegas, C. (1981).
\newblock Inner statistical inference {II}.
\newblock {\em Annals of Statistics}, 9(4):768--776.

\bibitem[Zidek, 1969]{zidek69}
Zidek, J. (1969).
\newblock A representation of {B}ayesian invariant procedures in terms of
  {H}aar measure.
\newblock {\em Annals of the Institute of Statistical Mathematics},
  21(1):291--308.

\end{thebibliography}



\appendix

\section*{Appendix 1: Introduction to invariance in statistics}
\label{app_intro_inv1}

For a locally compact Hausdorff group $\G$, $K(\G)$ denotes the class of all continuous real-valued functions on $\G$ that have compact support. The left invariant Haar measure on $\G$ is defined as a Radon measure $H^l$ such that for all $f\in K(\G)$ and all $g_0\in\G$,
\begin{align}
\int_\G f(g) H^l(dg) = \int_\G f(g_0g) H^l(dg) \nonumber
\end{align}
The right invariant Haar measure $H^r$ on $\G$ is defined as $H^l$ but replacing $g_0g$ by $gg_0$. For a given group, both Haar measures exist and are unique up to multiplicative constants.

The (right) modulus $\Delta$ of $\G$ is the real positive valued function such that if $H^l$ is a left invariant Haar measure, then for all $f\in K(\G)$ and all $g_0 \in \G$,
\begin{align}
\label{eq_def_delta}
\int f(gg_0^{-1}) H^l(dg) = \Delta(g_0)\int f(g) H^l(dg) 
\end{align}
From the unicity of the Haar measure, $\Delta$ does not depend on the choice of $H^l$ and is a continuous function such that for all $g_1,g_2\in\G$, $\Delta(g_1g_2) = \Delta(g_1)\Delta(g_2)$, which implies that $\Delta(g^{-1}) = \Delta(g)^{-1}$. Note that for a group $\G$ the set of all right Haar measures is equal to the set of the left Haar measures if and only if $\Delta$ is identically equal to 1. This occurs for example when $\G$ is compact or commutative.

Concerning the Haar measures on the group $\G$, the initial definitions and properties imply that if $H^l$ is a left invariant Haar measure on $\G$ and $\Delta$ the modulus of $\G$ then for all $f\in K(\G)$
\begin{align}
\label{eq_prop_inverse_g}
\int f(g^{-1}) H^l(dg) = \int f(g) \Delta(g)^{-1} H^l(dg)
\end{align}
The modulus also enables to relate right and left invariant Haar measures. From the last property, the measure defined by 
\begin{align} 
\label{eq_multiplicateur2}
H^r(dg) = \Delta(g)^{-1} H^l(dg) 
\end{align}
is a right invariant Haar measure on $\G$. The same way, if $H^r$ is a right invariant Haar measure on $\G$, then the measure defined by $H^l(dg) = \Delta(g) H^r(dg)$ is a left invariant Haar measure.


The Haar measure is applied to statistics through the concept of invariance of a data model under a group of transformations. 
A parametric family $\mathcal{P}_\Theta=\big\{p(.|\theta),\theta\in\Theta\big\}$ of densities with respect to any measure $\mu$ on $\X$ is said to be invariant under the transformations group $\G$ if for each $g\in\G$ there exists a unique $\theta^*\in\Theta$ such that if the distribution of $X$ has the density $p(.|\theta)\in\mathcal P_\Theta$ then $Y = gX$ has the density $p(.|\theta^*)\in\mathcal P_\Theta$. This property defines the action of $\G$ on $\Theta$: $\theta^*$ may simply be denoted $\theta^*=\bar g\theta$ where $\{\bar g,g\in \G\}$ defines a group. 

A measure $\mu$ on $\X$ is said to be relatively invariant with multiplier $\chi$ under the group $\G$ if for all $ f\in K(\mathcal X)$ and $g\in\G$
\begin{align}
\int f(x) \mu(dx) = \chi(g) \int f(gx) \mu(dx)
\end{align}
If we assume that both the family of densities and the measure $\mu$ are respectively invariant and relatively invariant, schematically we get 
$p(x|\theta) = \chi(g) p(gx |g\theta)$ for all $x\in\mathcal X, \theta\in\Theta$ and $g\in\G$. For more about the connection between such a multiplier and the Jacobian of the transformation that leads to $gx$ from $x$, see for example \citet{berger85} or \citet{eaton07}. 
Note that the theorem \ref{th1} could be formulated differently, by defining the invariance of a probability model, but this phrasing is less common than the invariance of a family of probability densities and this would have entailed a longer presentation.

To shorten the preliminaries and without assuming any knowledge about group theory, we will not refer to group properties like transitivity, orbits... and will concretely simply assume that $\Theta$ and $\G$ are isomorphic. More precisely, we will assume that the transformation $\phi_{\theta} : \G \mapsto \Theta$ with $\phi_\theta(g) = g\theta$ is one-to-one whatever $\theta\in\Theta$. The right Haar prior on $\Theta$ is to be induced from the right Haar measure $H^r$ on $\G$ and the action of $\G$ on $\Theta$. From the frame chosen, the right Haar prior $\Pi_{a}^r$ is simply defined by $\Pi_{a}^r = H^r(\phi_{a}^{-1})$, with $a\in\Theta$. As shown in \citet{villegas81}, it turns out that the measure $\Pi_{a}^r$ actually does not depend on $a$. The induced prior is therefore unique for a fixed $H^r$ and noted $\Pi^r$. $\Pi^r = H^r(\phi_{a}^{-1})$ means that for any measurable subset $A\subset \Theta$, $\Pi^r(A) = H^r \big(\phi_{a}^{-1}A\big)$ with $\phi_{a}^{-1}A = \big\{\phi_{a}^{-1}\theta | \theta\in A\big\}$. Note that a subset $A=d\theta$ denotes an infinitesimal subset centered around $\theta$, where $\theta$ is implicit. $\Pi^r$ can be normalized into a probability measure if and only if the group $\G$ is compact, and in this case we can go back to the usual notation $\Pi^r(A) = \Pr(\theta\in A)$ where the measure $\Pi^r$ is implicit in $\Pr(.)$. 

Finally, from the data model density $p(.|\theta)$ and the prior $\Pi^r$, the posterior measure $\Pi^r_{x}$ on $\Theta$ is classically defined by
\begin{align}
\label{eq_posteri}
\Pi^r_{x}(B) &= \frac{\int_B p(x|\theta) \Pi^r(d\theta)}{m(x)} ~~~~~~~\text{ for all} ~~ B\subset\Theta \\
\text{with } ~m(x) &= \int p(x|\theta) \Pi^r(d\theta) \nonumber
\end{align}
where the marginal $m(x)$ density of $x$ is always assumed to be finite, so that $\Pi^r_{x}$ defines a probability measure even if $\Pi^r$ does not. Then the posterior probability of an event is denoted by $\Pr(.|x)$, meaning $\Pr(\theta\in B|x) = \Pi^r_{x}(B)$.

\section*{Appendix 2: General theorem and its proof}

\begin{theorem}
\label{th1}
Call $\mathcal P_\Theta = \{p(.|\theta),\theta\in\Theta\}$ a family of probability densities with respect to a measure $\mu^r$ on $\mathcal X$, specified later, and call $\G$ a group acting on $\X$. Assume that $\mathcal P_\Theta$ is invariant under the action of the group $\G$ on $\mathcal X$ and note $\bar g\theta$ the induced action of the element $g\in\G$ on the element $\theta\in\Theta$. Call $H^r$ any right Haar measure of $\G$ and define the transformations $\phi_\theta$ (for $\theta\in\Theta$) and $\phi_{x}$ (for $x\in\X$) by
\begin{equation}
\label{eq_def_phi_x} 
\begin{array}{lr}
\begin{array}{rl}
\phi_{\theta} : \bar\G &\mapsto \Theta\\
    \bar g &\mapsto \bar g\theta 
\end{array}
~~~~~~&~~~~~~
\begin{array}{rl}
\phi_{x} : \G &\mapsto \mathcal X \\
   g &\mapsto gx
\end{array}
\end{array}
\end{equation}
Assume that 
\begin{enumerate}
\item $\phi_{\theta}$ is one-to-one for all $\theta\in\Theta$ and $\phi_{x}$ is one-to-one for all $x\in \mathcal X$.
\item The prior measure $\Pi^r$ on $\Theta$ is the measure induced by $H^r$ via $\phi_{\theta}$ and the measure $\mu^r$ on $\mathcal X$ is the measure induced by $H^r$ via $\phi_{x}$: $\Pi^r = H^r(\phi_{\theta}^{-1})$ and $\mu^r = H^r(\phi_{x}^{-1})$.
\item The marginal density of $x$ is finite, so that the posterior measure $\Pi^r_{x}$ on $\Theta$, classically defined by the equation (\ref{eq_posteri}), defines the posterior probability $\Pr(.|x_0)$.
\end{enumerate}
Then, the PLR defined by the equation (\ref{def_PLR}) can be reexpressed, for any $\zeta>0$ and {\em any} $c\in\mathcal X$, as the frequentist integral:
\begin{align}
\PLR(x_0,\zeta) &= \Pr\Big(~p(x_0|\theta_0)\Delta\big(\phi_{x_0}^{-1}c\big) ~ \le ~ \zeta~ p(x|\theta_0)\Delta\big(\phi_{x}^{-1}c\big) \mid \theta_0\Big)
\end{align}
where $\Delta$ is the modulus of the group $\G$, as defined in the equation (\ref{eq_def_delta}), and in practice $x_0\in\mathcal X$ is the observed data and $\theta_0\in\Theta$ the parameter value under the null hypothesis.
\end{theorem}

Note that as seen in the previous appendix, the measures $\mu$ and $\Pi^r$ defined in the theorem \ref{th1} do not depend on the choice of $\theta\in\Theta$ and $x\in\mathcal X$ in the functions $\phi_\theta$ and $\phi_x$. In order to clarify the proof, we note $a$ instead of $x$ and $b$ instead of $\theta$ in the following. We shall make use of the following lemma:
\begin{lemma}
\label{lemme_mu_relativ_inv}
The measures $\mu$ on $\mathcal X$ and $\Pi^r$ on $\Theta$ induced above by the right Haar measure $H^r$ on $\G$ are relatively invariant with modulus $\Delta^{-1}$. 
\end{lemma}

\begin{proof}
\begin{align*}
\int f(g_0x) \mu(dx) &= \int f(g_0x) H^r\phi_{a}^{-1}(dx) ~~\text{(Def. of $\mu$ in the Cond. of Th. \ref{th1})}\\
  &= \int f\big(g_0\phi_{a}g\big) H^r(dg) ~~\text{(transformation $g=\phi_{a}^{-1}x$)} \\
  &= \int f\big(g_0ga\big) \Delta(g)^{-1} H^l(dg) ~~\text{(Def. of $\phi_{a}$ and prop. eq. (\ref{eq_multiplicateur2}))} \\
  &= \Delta(g_0) \int f\big(g_0ga\big) \Delta(g_0g)^{-1} H^l(dg) ~~\text{(Multiplicity prop. of $\Delta$)} \\
  &= \Delta(g_0) \int f(ga) \Delta(g)^{-1} H^l(dg) ~~\text{($H^l$ left invariant)} \\
  &= \Delta(g_0) \int f(x) \mu(dx) ~~\text{(previous computation made in reverse order)}
\end{align*}

This also implies that a Haar prior induced as in the theorem \ref{th1}, i.e. from a right invariant Haar measure on $\G$, is relatively invariant.

\begin{align}
\PLR(x_0,\zeta) &= \Pr\Big(p(x_0|\theta_0)\le\zeta p(x_0|\theta)~\Big|~x_0\Big) \nonumber \\
  &= \frac{1}{m(x_0)} \int_{\big\{\theta \mid p(x_0|\theta_0)\le\zeta p(x_0|\theta)\big\}} p(x_0|\theta)\Pi^r(d\theta) \label{eq_plr_cal1}\\
  &= \frac{1}{m(x_0)} \int_{\big\{\theta \mid p(x_0|\theta_0)\le\zeta p(x_0|\theta)\big\}} p(x_0|\theta)H^r\big(\phi_{b}^{-1}(d\theta)\big)~~\text{(Def. $\Pi^r$ in the Cond. of Th. \ref{th1})} \nonumber\\
  &= \frac{1}{m(x_0)} \int_{\big\{g \mid p(x_0|\theta_0)\le\zeta p(x_0|\phi_{b}g)\big\}} p(x_0|\phi_{b}g)H^r(dg) ~~\text{($g=\phi_{b}^{-1}\theta$)} \nonumber\\
  &= \frac{1}{m(x_0)} \int_{\big\{g \mid p(x_0|\theta_0)\le\zeta p(x_0|gb)\big\}} p(x_0|gb)H^r(dg)~~\text{(Def. $\phi_{\theta}$ eq. (\ref{eq_def_phi_x}))} \nonumber\\
  &= \frac{1}{m(x_0)} \int_{\big\{g \mid p(x_0|\theta_0)\le\zeta p(x_0|gb)\big\}} p(x_0|gb)\Delta(g)^{-1}H^l(dg)~~\text{(Prop. eq. (\ref{eq_multiplicateur2}))}\nonumber \\
  &= \frac{1}{m(x_0)} \int_{\big\{g \mid p(x_0|\theta_0)\le\zeta p(x_0|g^{-1}b)\big\}} p\big(x_0\big|g^{-1}b\big)H^l(dg) ~~\text{(Prop. eq. (\ref{eq_prop_inverse_g}))} \nonumber
\end{align}

But according to lemma \ref{lemme_mu_relativ_inv}, $\mu$ is relatively invariant with modulus $\Delta^{-1}$. Since the density family is invariant, 
\begin{align*}
p(x|\theta) = \Delta(g)^{-1} p(gx |g\theta) ~~~ \text{for all } x\in\mathcal X, \theta\in\Theta, g\in\G
\end{align*}
i.e. 
\begin{align*}
p(x|g^{-1}\theta') = \Delta(g)^{-1} p(gx |\theta') ~~ \text{for all } x\in\mathcal X, \theta'\in\Theta, g\in\G
\end{align*}

Then,
\begin{align*}
\PLR(x_0,\zeta) &= \frac{1}{m(x_0)} \int_{\big\{g :\mid p(x_0|\theta_0)\le\zeta p(gx_0|b)\Delta(g)^{-1}\big\}} \Delta(g)^{-1} p(gx_0|b)H^l(dg) \\
  &= \frac{1}{m(x_0)} \int_{\big\{g \mid p(x_0|\theta_0)\le\zeta p(gx_0|b)\Delta(g)^{-1}\big\}} p(gx_0|b)H^r(dg)~~\text{(Prop. eq. (\ref{eq_multiplicateur2}))} \\
  &= \frac{1}{m(x_0)} \int_{\big\{g \mid p(x_0|\theta_0)\le\zeta p(gx_0|b)\Delta(g)^{-1}\big\}} p(gx_0|b)\mu\big(\phi_{a}(dg)\big)~~\text{(Def. $\mu$)}
\end{align*}

It can be noticed that the equation (\ref{eq_plr_cal1}) depends neither on $a\in\mathcal X$ nor on $b\in\Theta$. Choose now for simplicity $a=x_0$. Then, making the transformation $x=\phi_{x_0}g=gx_0$,
\begin{align*}
\PLR(x_0,\zeta)  &= \frac{1}{m(x_0)} \int_{\big\{x \mid p(x_0|\theta_0)\le\zeta p(x|b)\Delta\big(\phi_{x_0}^{-1}x\big)^{-1}\big\}} p(x|b)\mu(dx)
\end{align*}

By a similar computation we get the expression of the marginal density of $X$ evaluated at $x_0$:
\begin{align*}
m(x_0) &= \int p(x_0|\theta)\Pi^r(d\theta) = \int p(x|b)\mu(dx) = 1
\end{align*}
The marginal density of $X$ is constant, the same way the frequentist risk of an invariant estimator does not depend on $\theta$. So
\begin{align*}
\PLR(x_0,\zeta) &= \int_{\big\{x \mid p(x_0|\theta_0)\le\zeta p(x|b)\Delta\big(\phi_{x_0}^{-1}x\big)^{-1}\big\}} p(x|b)\mu(dx)
\end{align*}

In order to get a form closer to a p-value, we choose from now $b=\theta_0$ and note that for {\em any} $c\in\mathcal X$,
\begin{align}
\Delta\big(\phi_{x_0}^{-1}x\big) = \frac{\Delta\big(\phi_{c}^{-1}x\big)}{\Delta\big(\phi_{c}^{-1}x_0\big)}  \label{eq_role_c}
\end{align}
because if we note
\begin{align*}
g &= \phi_{x_0}^{-1}x \\
g_1 &= \phi_{c}^{-1}x \\
g_2 &= \phi_{c}^{-1}x_0 
\end{align*}
then on one side $gx_0 = x$ and on the other $g_1(g_2^{-1}x_0) = g_1c = x$ so that
\begin{align*}
gx_0 &= (g_1g_2^{-1})x_0 \\
\text{so }~ \phi_{x_0}g &= \phi_{x_0}(g_1g_2^{-1}) \\
\text{so }~~~~~ g &= g_1g_2^{-1} ~~~ \text{($\phi_{a}$ is one-to-one)} \\
\text{so }~ \Delta(g) &= \frac{\Delta(g_1)}{\Delta(g_2)} ~~~ \text{(Prop. of $\Delta$)}
\end{align*}

Finally, for any $c\in\mathcal X$ 
\begin{align}
\PLR(x_0,\zeta) &= \int_{\left\{x \mid \frac{p(x_0|\theta_0)}{\Delta(\phi_{c}^{-1}x_0)}\le\zeta \frac{p(x|\theta_0)}{\Delta(\phi_{c}^{-1}x)}\right\}} p(x|\theta_0)\mu(dx)
\end{align}

It is also interesting to note that
\begin{align}
\phi_{a}^{-1}b &= (\phi_{b}^{-1}a)^{-1} \label{eq_phi_ab} \\
\text{since } g=\phi_{a}^{-1}b \Rightarrow ga = b &\Rightarrow a=g^{-1}b \Rightarrow g^{-1} = \phi_{b}^{-1}a \nonumber
\end{align}
so that the same way we have
\begin{align*}
\PLR(x_0,\zeta) &= \int_{\big\{x \mid p(x_0|\theta_0)\Delta\big(\phi_{x_0}^{-1}c\big)\le\zeta p(x|\theta_0)\Delta\big(\phi_{x}^{-1}c\big)\big\}} p(x|\theta_0)\mu(dx) \\
  &= \Pr\Big(p(x_0|\theta_0)\Delta\big(\phi_{x_0}^{-1}c\big)\le\zeta p(x|\theta_0)\Delta\big(\phi_{x}^{-1}c\big) \Big| \theta_0\Big)
\end{align*} 
\end{proof}

\section*{Appendix 3: Proof of the theorem \ref{th_lebesgue} and corollaries \ref{th1_suff}, \ref{cor_pval}}

The theorem \ref{th_lebesgue} is a corollary of the theorem \ref{th1} presented and proved in the previous appendix: the theorem \ref{th1} can be reexpressed more simply by assuming that the likelihood family and the induced Haar measures are absolutely continuous with respect to the Lebesgue measure. 

\begin{proof}[Proof of the theorem \ref{th_lebesgue}]
The proof  only consists of reexpressing the domains of integration because the integrands expression are not functions of the use of the decomposition of the measures over some other measures ($\mu$ or Lesbesgue). The proof even actually only consists of reexpressing the domain of integration of the p-value because the domain of integration of the PLR does not depend on the density over $\mathcal X$ used since the domain of integration is a subset of $\Theta$, not $\mathcal X$. 

If we note $p^{\mu}(.|\theta)$ the density with respect to the induced Haar measure $\mu^r$ and $p(.|\theta)$ the density with respect to the Lebesgue measure, we have by definition
\begin{align*}
P(dx|\theta) = p^{\mu}(x|\theta)\mu^r(dx) =  p^{\mu}(x|\theta)\pi^r(x)dx ~~&\text{ and }~~ P(dx|\theta) = p(x|\theta)dx \\
\text{ and so }~~ p^{\mu}(x|\theta) = \frac{p(x|\theta)}{\pi^r(x)}
\end{align*}
On the other side the modulus $\Delta$ can also be reexpressed as a function of the induced prior densities $\pi^l(x)$ and $\pi^r(x)$. From the equations (\ref{eq_phi_ab}) and (\ref{eq_multiplicateur2}),
\begin{align*}
\Delta\left(\phi_x^{-1}c\right) &= \Delta\left(\phi_c^{-1}x\right)^{-1} = \frac{H^l\left(d\phi_c^{-1}x\right)}{H^r\left(d\phi_c^{-1}x\right)} = \frac{\mu^r(dx)}{\mu^l(dx)} = \frac{\pi^r(x)}{\pi^l(x)}
\end{align*}
Combining these two results we get
\begin{align*}
p^{\mu}(x|\theta) \Delta\left(\phi_x^{-1}c\right) =  \frac{p(x|\theta)}{\pi^l(x)}
\end{align*}
\end{proof}

\begin{proof}[Proof of the  corollary \ref{th1_suff}]
\begin{align*}
\PLR(x_0,\zeta) &= \Pr\Big((p_{X|\theta_0}(x_0) ~\le~ \zeta~ p_{X|\theta}(x_0)~\Big|~x_0\Big) \\
  &= \Pr\Big(p_{X|S(X)}(x_0|S(x_0))~p_{S(X)|\theta_0}(S(x_0)) ~\le~ \zeta~p_{X|S(X)}(x_0|S(x_0))~p_{S(X)|\theta}(S(x_0))~\Big|~x_0\Big) 
\end{align*}
because since $S(x)$ is a function of $x$, $p_{X|\theta}(x) = p_{X,S(X)|\theta}\big(x,S(x)\big)$ and since in addition $S(X)$ is a sufficient statistic of $X$,
\begin{align}
p_{X|\theta}(x) &= p_{X |S(X),\theta}\big(x|S(x)\big)~p_{S(X)|\theta}\big(S(x)\big) \nonumber \\
  &= p_{X|S(X)}\big(x|S(x)\big)~~p_{S(X)|\theta}\big(S(x)\big)
\end{align}

Simplifying the densities which do not depend on $\theta$,
\begin{align*}
\PLR(x_0,\zeta) &= \Pr\Big(p_{S(X)|\theta_0}\big(S(x_0)\big) ~\le~ \zeta ~p_{S(X)|\theta_0}\big(S(x_0)\big)~\Big|~S(x_0)\Big) ~~\text{ if ~ $p_{X|S(X)}\big(x_0|S(x_0)\big)>0$}\\
  &= \Pr\Big(p_{S(X)|\theta_0}\big(S(x_0)\big)(\pi^l(S(x_0)))^{-1}~\le~\zeta~ p_{S(X)|\theta_0}\big(S(x)\big)(\pi^l(S(x)))^{-1} \Big| \theta_0\Big) ~~~ \text{ (Th. \ref{th_lebesgue}) }
\end{align*}
\end{proof}

\begin{proof}[Proof of the  corollary \ref{cor_pval}] 
First reexpress the PLR under the conditions of the theorem \ref{th_lebesgue} by using a cumulative distribution. Note $T(x)$ the statistic: 
\begin{align}
T(x) &= p_{S(X)|\theta_0}\big(S(x)\big)(\pi^l(S(x)))^{-1}  \nonumber
\end{align}
Seen as a random variable, the dataset $x$ induces the random variable $T(X)$ the same way the statistic $S(x)$ induced $S(X)$. Note $F_{T(X)|\theta_0}$ the cumulative distribution of $T(X)$ under the null hypothesis:
\begin{align*}
F_{T(X)|\theta_0}(\zeta) &= \Pr\big(T(x)\le\zeta \big| \theta_0\big)
\end{align*}
Starting from the theorem \ref{th_lebesgue}, the PLR can be reexpressed as
\begin{align}
\PLR(x_0,\zeta) &= 1 - F_{T(X)|\theta_0}\big(\zeta^{-1}T(x_0)\big) \nonumber
\end{align}

In particular, for a threshold $\zeta=1$, one can directly notice that the PLR is equal to the \textit{p-value} defined for the GLR by the equation (\ref{eq_pval_f}), but now instead associated to the test statistic $T(x)$. 

Also note that the frequentist test corresponding to the PLR is then given, for any threshold $\lambda>0$, by
\begin{align}
\text{Reject $H_0$ if }~ & ~~~~~p_S\big(S(x)|\theta_0\big)~(\pi^l(S(x)))^{-1}~ \le~ \lambda \nonumber
\end{align}
\end{proof}

\section*{Appendix 4: Proof of the remark \ref{lem_joint_meas}}

If there exists a joint measure over $\Theta_0\times\Theta_1\times\X | H_0$ and if the events defined on the sets $\Theta_0\times\X | H_0$ and $\Theta_1 | H_0$ are independent, then
\begin{align}
P(d\theta_0,d\theta_1,x | H_0) &= P(d\theta_0,x | H_0) P(d\theta_1 | H_0)\nonumber
\end{align}
This can be reformulated using the standard previous notations:
\begin{align}
P(d\theta_0,d\theta_1,x | H_0) &= \Pi_0(d\theta_0|x) m_0(x) \Pi_1(d\theta_1)\nonumber
\end{align}
Then, noting $\Pi_{01,0}(. | x) = P(. | H_0,x)$,
\begin{align}
\Pi_{01,0}(d\theta_0,d\theta_1 | x) &= \frac{P(d\theta_0,d\theta_1,x|H_0)}{\int P(d\theta_0,d\theta_1,x|H_0)} \nonumber\\
  &= \frac{m_0(x)\Pi_0(d\theta_0|x) \Pi_1(d\theta_1)}{m_0(x)}\nonumber \\
  &= \Pi_0(d\theta_0|x) \Pi_1(d\theta_1)\nonumber
\end{align}

\section*{Appendix 5: Proof of the proposition \ref{lemma_np_b}}

Recall that the set $\mathcal{R}^{\ast}(x)$ defined in equation (\ref{eq_set_lr}) is the LR set that rejects $H_0$, and upon which is defined $\PLR_1$ in equation (\ref{eq_plr1}). Call $\PFA_B(\mathcal{R}^{\ast},x)$ and $\PD_B(\mathcal{R}^{\ast},x)$ the associated integrals defined in equations (\ref{eq_pfab}) and (\ref{eq_pdb}). Call $\mathcal{R}(x) \subset \Theta_0\times\Theta_1$ any other set and $\PFA_B(\mathcal{R},x)$ and $\PD_B(\mathcal{R},x)$ its associated integrals.

The goal is to show that $\PFA_B(\mathcal{R},x)\le\PFA_B(\mathcal{R}^{\ast},x)$ implies that $\PD_B(\mathcal{R},x)\le\PD_B(\mathcal{R}^{\ast},x)$ for any test set $\mathcal{R}$. The fact that $\PD_B(\mathcal{R},x)\le\PD_B(\bar{\mathcal{R}}^{\ast},x)$ implies that $\PFA_B(\mathcal{R},x)\le\PFA_B(\bar{\mathcal{R}}^{\ast},x)$ is shown in a reciprocal way.

One can check that the following inequality holds for all $x\in\X$, all $\theta_0\in\Theta_0$ and all $\theta_1\in\Theta_1$:
\begin{align}
\left(I_{\mathcal{R}^{\ast}(x)}(\theta_0,\theta_1) - I_{\mathcal{R}(x)}(\theta_0,\theta_1)\right)(p(x|\theta_0) - \zeta p(x|\theta_1)) \le 0\nonumber
\end{align}

Since the inequality is true for all $x, \theta_0$ and $\theta_1$, we can multiply the left hand side by any positive term and integrate over $\Theta_0\times\Theta_1$. This implies in particular:
\begin{align}
\int_{\Theta_0}\int_{\Theta_1} \frac{\Pi_0(d\theta_0)}{m_0(x)} \frac{\Pi_1(d\theta_1)}{m_1(x)} \left(I_{\mathcal{R}^{\ast}(x)}(\theta_0,\theta_1) - I_{\mathcal{R}(x)}(\theta_0,\theta_1)\right)(p(x|\theta_0) - \zeta p(x|\theta_1)) \le 0\nonumber
\end{align}

But since $\Pi_i(d\theta_i|x) = \Pi_i(d\theta_i)p(x|\theta_i)m_i(x)^{-1}$ for $i=0,1$, this implies
\begin{align}
&\frac{1}{m_1(x)}\int_{\Theta_0}\int_{\Theta_1} \Pi_0(d\theta_0|x) \Pi_1(d\theta_1) \left(I_{\mathcal{R}^{\ast}(x)}(\theta_0,\theta_1) - I_{\mathcal{R}(x)}(\theta_0,\theta_1)\right)\nonumber \\
& - \frac{\zeta}{m_0(x)}\int_{\Theta_0}\int_{\Theta_1} \Pi_1(d\theta_1|x) \Pi_0(d\theta_0) \left(I_{\mathcal{R}^{\ast}(x)}(\theta_0,\theta_1) - I_{\mathcal{R}(x)}(\theta_0,\theta_1)\right)\le 0\nonumber
\end{align}
where we recognize $\PFA_B$ and $\PD_B$ as defined in equations (\ref{eq_pfab}) and (\ref{eq_pdb}):
\begin{align}
&\frac{\PFA_B(\mathcal{R}^{\ast}(x)) - \PFA_B(\mathcal{R}(x))}{m_1(x)} - \zeta\frac{\PD_B(\mathcal{R}^{\ast}(x)) - \PD_B(\mathcal{R}(x))}{m_0(x)} \le 0 \nonumber
\end{align}
Therefore we finally have
\begin{align}
\zeta\frac{\PD_B(\mathcal{R}(x)) - \PD_B(\mathcal{R}^{\ast}(x))}{m_0(x)} \le \frac{\PFA_B(\mathcal{R}(x)) - \PFA_B(\mathcal{R}^{\ast}(x))}{m_1(x)} \nonumber
\end{align}
from which we conclude the final implication 
\begin{align}
\PFA_B(\mathcal{R}(x)) \le \PFA_B(\mathcal{R}^{\ast}(x))  ~~ \Rightarrow ~~ \PD_B(\mathcal{R}(x)) \le \PD_B(\mathcal{R}^{\ast}(x))\nonumber
\end{align}

The fact that $\PD_B(\mathcal{R},x)\le\PD_B(\bar{\mathcal{R}}^{\ast},x)$ implies that $\PFA_B(\mathcal{R},x)\le\PFA_B(\bar{\mathcal{R}}^{\ast},x)$ is shown in a reciprocal way.

\end{document}